%% file: paper.tex
\documentclass[onefignum, onetabnum]{siamart190516}

\usepackage[linesnumbered, ruled, vlined, algo2e]{algorithm2e}
\input{shared}

\title{Computing solution space properties of combinatorial optimization problems via generic tensor networks}

\begin{document}

\maketitle

\begin{abstract}
We introduce a unified framework to compute the solution space properties of a broad class of combinatorial optimization problems. These properties include finding one of the optimum solutions, counting the number of solutions of a given size, and enumeration and sampling of solutions of a given size. Using the independent set problem as an example, we show how all these solution space properties can be computed in the unified approach of generic tensor networks. We demonstrate the versatility of this computational tool by applying it to several  examples, including computing the entropy constant for hardcore lattice gases, studying the overlap gap properties, and analyzing the performance of quantum and classical algorithms for finding maximum independent sets.
\end{abstract}



\section{Introduction}

An important class of problems in graph theory and combinatorial optimization  can be formulated as satisfiability problems involving constraints specified over a vertex and its neighborhood.
These problems include, for example, the independent set problem, the cutting problem, dominating set, set packing, set covering, vertex coloring, K-SAT,  the clique problem, and the vertex cover problem~\cite{Moore2011}.
These problems have a wide range of applications in scheduling, logistics,
wireless networks and telecommunication, and computer vision, among others~\cite{Butenko2003, Wu2015}.
Finding an optimum solution for these problems is typically NP-hard in the worst case~\cite{Hastad1996}.

In this Article, we introduce a unified framework to compute a broad class of properties 
associated with the solutions of these problems, beyond just finding an optimum solution.
We call them \textit{solution space properties}. In practice, these can be much harder to compute (corresponding e.g.\ to \#P-complete class~\cite{Moore2011}). However, these properties can be crucial for understanding detailed properties of hard combinatorial optimization problems.
For example, for the independent set problem, these \textit{solution space properties} can include not only the maximum or minimum independent set size but also the number of sets at a given size, enumeration of all sets at a given size, and direct sampling of such sets when they are too large to be fit into memory.
They can be used to understand the hardness of finding an optimum solution for a given problem instance and the performance of a specific solver.  
For example, the number of configurations at different sizes can inform how likely a simulated annealing algorithm will be trapped in local minima at certain sizes~\cite{Xu2018}.
The pair-wise Hamming distance distribution of configurations at a given size can indicate the presence or absence of the overlap gap property~\cite{Gamarnik2013, Gamarnik2019}, which can be used to bound the performance of local optimization algorithms.
In a recent experiment based on a Rydberg atom array quantum computer, 
the counting and the configuration space connectivity information was used to find maximum independent set (MIS) problem instances that are hard for simulated annealing and to evaluate the corresponding quantum algorithm performance~\cite{Ebadi2022}.
The need to understand these important aspects of combinatorial optimization motivates us to find methodologies to compute these solution space properties.

To this end, we show how to obtain all of these seemingly unrelated properties in a unified approach using
\textit{generic tensor networks}. Tensor networks are a computational model widely used in condensed matter physics~\cite{Orus2014}, quantum computing~\cite{Markov2008}, big data~\cite{Cichocki2014}, mathematics~\cite{Oseledets2011} and combinatorial optimization~\cite{Biamonte2015, Biamonte2017,Kourtis2019}.
They are also known as the sum-product networks in probabilistic modeling~\cite{Bishop2006} or \texttt{einsum} in linear algebra libraries such as NumPy~\cite{Harris2020}.
Recent progress in simulating quantum circuits with tensor networks~\cite{Gray2021, Pan2021, Kalachev2021} makes it possible to contract a randomly structured sparse tensor network with up to thousands of tensors in a reasonable time.
In previous studies, the data types of the tensor elements are typically restricted to standard number types such as real numbers and complex numbers.
Here, we extend to \textit{generic tensor networks} by generalizing the tensor element data types to any type that has the algebraic structure of a commutative semiring.
In what follows, for clarity of presentation, we focus on the independent set problem in the main text, and show how to compute the solution space properties for other combinatorial optimization problems in \Cref{sec:otherproblems}. The latter includes cutting, matching, vertex coloring, satisfiability,  dominating set,  set packing,  set covering, and the clique problem.

The paper is organized as follows.
We first introduce the basic concepts of tensor networks and generic programming in \Cref{sec:tn} and \Cref{sec:generic}.
Then we show how to reduce the independent set problem to a tensor network contraction problem in \Cref{sec:tnmap}. 
Subsequently, we explain how to engineer the element types to compute various solution space properties in \Cref{sec:counting}, \Cref{sec:enumeration}, and \Cref{sec:weighted}.
Lastly, we provide three example applications in \Cref{sec:examples} to demonstrate the versatility of our tool.
A benchmark to demonstrate the performance of our algorithms can be found in both the \Cref{sec:benchmark} and the code repository~\cite{GenericTensorNetworks}.

\section{Tensor networks}\label{sec:tn}
A tensor network is a multi-linear map from a collection of labelled tensors $\mathcal{T}$ to an output tensor.
It is formally defined as follows.
\begin{definition}[Tensor Network~\cite{Cirac2021, Orus2014}]
    A tensor network is a multi-linear map specified by a triple of $\mathcal{N} = (\Lambda, \Tensors, \boldsymbol{\sigma}_o)$,
    where $\Lambda$ is a set of symbols (or labels),
    $\Tensors = \{T^{(1)}_{\boldsymbol{\sigma}_1}, T^{(2)}_{\boldsymbol{\sigma}_2}, \ldots, T^{(M)}_{\boldsymbol{\sigma}_M}\}$ is a set of tensors as the inputs,
    and $\boldsymbol{\sigma}_o$ is a string of symbols labelling the output tensor.
    Each $T^{(k)}_{\boldsymbol{\sigma_k}} \in \mathcal{T}$ is labelled by a string $\boldsymbol{\sigma}_k \in \Lambda^{r \left(T^{(k)} \right)}$, where $r \left(T^{(k)} \right)$ is the rank of $T^{(k)}$.
    The multi-linear map or the \textbf{contraction} on this triple is
    \begin{equation}
        O_{\boldsymbol{\sigma}_o} = \sum_{\Lambda \setminus \sigma_o} \prod_{k=1}^{M} T^{(k)}_{\boldsymbol{\sigma_k}},
    \end{equation}
    where the summation runs over all possible configurations over the set of symbols absent in the output tensor.
\end{definition}
For example, the matrix multiplication can be specified as a tensor network
\begin{equation}
\mathcal{N}_{\rm matmul} = \left(\{i,j,k\}, \{A_{ij}, B_{jk}\}, ik\right),
\end{equation}
where $A_{ij}$ and $B_{jk}$ are input matrices (two-dimensional tensors), and $(i,k)$ are labels associated to the output.
The contraction is defined as $O_{ik} = \sum_j A_{ij}B_{jk}$, where the subscripts are for tensor indexing, and the tensor dimensions with the same label must have the same size.
The graphical representation of a tensor network is an open hypergraph that having open hyperedges, where an input tensor is mapped to a vertex and a label is mapped to a hyperedge that can connect an arbitrary number of vertices, while the labels appearing in the output tensor are mapped to open hyperedges.
Our notation is a minor generalization of the standard tensor network notation used in physics as we do not restrict the number of times a label can appear in the tensors to two. 
While this generalized form is equivalent in representation power, it can have smaller contraction complexity as will be illustrated in \Cref{sec:tensorbad}.

\begin{example}
\begin{equation}
\begin{split}
    \Lambda &= \{i,j,k,l,m\},\\
    \mathcal{T} &= \{A_{jkm}, B_{mil}, V_{jm}\},\\
    \boldsymbol{\sigma}_o &= ijk,\\
    \end{split}
\end{equation}
is a tensor network that can be evaluated as $O_{ijk} = \sum_{ml}A_{jkm} B_{mil} V_{jm}$.
Its hypergraph representation is shown below, where we use different colors to represent different hyperedges.

\vspace{1em}
\centerline{\begin{tikzpicture}[
    dot/.style = {circle, fill, minimum size=#1,
                inner sep=0pt, outer sep=0pt},
    dot/.default = 6pt  
                    ]  
    \def\dx{0};
    \def\r{0.5cm}
    \def\sr{0.15cm}
    \def\ax{0}
    \def\ay{0}
    \def\bx{1}
    \def\by{1}
    \def\cx{1}
    \def\cy{-1}
    \node[color=white,fill=black,dot=\r] at (\ax+\dx,\ax) (a) {A};
    \node[color=white,fill=black,dot=\r] at (\bx+\dx,\by) (b) {B};
    \node[color=white,fill=black,dot=\r] at (\cx+\dx,\cy) (v) {V};
    \node[color=transparent,draw=transparent,dot=0] at (\ax-1,\by) (o1) {};
    \node[color=transparent,draw=transparent,dot=0] at (\ax-1.5,\ay) (o2) {};
    \node[color=transparent,draw=transparent,dot=0] at (\ax-1,\cy) (o3) {};
    \node at (\ax-0.8,\ay) (k) {k};
    \node at (\bx+0.4,\ay) (m) {m};
    \node at (\ax,\cy) (j) {j};
    \node at (\bx+1,\by) (l) {l};
    \node at (\bx-1,\by) (i) {i};
    \draw[color=blue,thick] (i) -- (b);
    \draw[color=blue,thick] (i) -- (o1);
    \draw[color=cyan,thick] (l) -- (b);
    \draw[color=violet,thick] (k) -- (a);
    \draw[color=violet,thick] (k) -- (o2);
    \draw[color=black,thick] (b) -- (m);
    \draw[color=black,thick] (m) -- (a);
    \draw[color=black,thick] (m) -- (v);
    \draw[color=red,thick] (a) -- (j);
    \draw[color=red,thick] (v) -- (j);
    \draw[color=red,thick] (o3) -- (j);
\end{tikzpicture}}
\end{example}

\section{Generic programming tensor contractions}\label{sec:generic}
In previous works relating tensor networks and combinatorial optimization problems~\cite{Kourtis2019, Biamonte2017}, the element types in the tensor networks are limited to standard number types such as floating-point numbers and integers.
We propose to use more general element types with a certain algebraic property.
With different data types, we can solve different problems within the same unified framework.
This idea of using the same program for different purposes is also called generic programming in computer science:

\begin{definition}[Generic programming ~\cite{Stepanov2014}]
   Generic programming is an approach to programming that focuses on designing algorithms and data structures so that they work in the most general setting without loss of efficiency.
\end{definition}
This definition of generic programming covers two major aspects: ``work in the most general setting'' and ``without loss of efficiency''.
By the most general setting, we mean that a single program should work correctly for the most general input data types. 
For example, suppose we want to write a function that raises an element to a power, $f(x, n) := x^n$.
One can easily write a function for standard number types that computes the power of $x$ in $O \left( \log(n) \right)$ steps using the multiply and square trick.
Generic programming does not require $x$ to be a standard number type;
instead, it treats $x$ as an element with an associative multiplication operation $\odot$ and a multiplicative identity $\mymathbb{1}$.
Then, when the program takes a matrix as an input instead of a standard number type, it computes the matrix power correctly without rewriting the program.
The second aspect is about efficiency. For dynamically typed languages such as Python, the type information is not available for type-specific optimizations at the compilation stage.
Therefore, one can easily write code that works for general input types, but the efficiency is not guaranteed; for example, the speed of computing the matrix multiplication between two NumPy arrays with Python objects as elements is much slower than statically typed (i.e. the type information can be accessed at the compilation stage) languages such as C++ and Julia~\cite{Bezanson2012}.
C++ uses templates for generic programming while Julia takes advantage of just-in-time compilation and multiple dispatches.
When these languages ``see'' a new input type, the compiler recompiles the generic program for the new type to generate an efficient binary.
A myriad of optimizations can be done during the compilation. For example, the compiler can optimize the memory layout of immutable elements with fixed sizes in an array to speed up array indexing.
In Julia, if a type is immutable and contains no references to other values, an array of that type can even be compiled to graphics processing units (GPU) for faster computation~\cite{Besard2018}.

This motivates us to identify the most general tensor element type allowed in a tensor network contraction program.
We find that as long as the tensor elements are members of a commutative semiring, the tensor network contraction will be well defined and the result will be independent of the contraction order.
In contrast with a field, a commutative semiring does not need not to have an additive inverse and a multiplicative inverse.
Giving up these nice properties of fields has significant implications for tensor computation: tensor network compression algorithms might not be applicable because matrix factorization is NP-hard for commutative semirings~\cite{Shitov2014} and matrix multiplication faster than $O(n^3)$ does not exist for an algebra without an additive inverse~\cite{Kerr1970}.
Here, we only use the commutative properties of an algebra for optimizing the tensor network contraction order.
To define a commutative semiring with the addition operation $\oplus$\footnote{We use the $\oplus$ operator throughout this paper to denote the generic addition, which is not the logical \texttt{XOR} operation, in which the symbol typically represents in computer science.} and the multiplication operation $\odot$ on a set $S$, the following relations must hold for any arbitrary three elements $a, b, c \in S$.
\begin{align*}
(a \oplus b) \oplus c = a \oplus (b \oplus c) & \hspace{5em}\text{$\triangleright$ commutative monoid $\oplus$ with identity $\mymathbb{0}$}\\
a \oplus \mymathbb{0} = \mymathbb{0} \oplus a = a &\\
a \oplus b = b \oplus a &\\
&\\
(a \odot b) \odot c = a \odot (b \odot c)  &   \hspace{5em}\text{$\triangleright$ commutative monoid $\odot$ with identity $\mymathbb{1}$}\\
a \odot  \mymathbb{1} =  \mymathbb{1} \odot a = a &\\
a \odot b = b \odot a &\\
&\\
a \odot (b\oplus c) = a\odot b \oplus a\odot c  &  \hspace{5em}\text{$\triangleright$ left and right distributive}\\
(a\oplus b) \odot c = a\odot c \oplus b\odot c &\\
&\\
a \odot \mymathbb{0} = \mymathbb{0} \odot a = \mymathbb{0}
\end{align*}
In the following sections, we will show how to compute solution space properties of independent sets using the same tensor network contraction algorithm by engineering tensor element algebra.
The Venn diagram in \Cref{fig:venn-diagram} shows the different types of algebra we will introduce in the main text and their relation, and \Cref{tbl:generictypes} summarizes which solution space properties can be computed by which tensor element types.

\begin{figure}[th]
   \centering
\centerline{\begin{tikzpicture}[]
    \scriptsize
    \node[draw,fill=lime!80,fill opacity=1, text opacity=1.0,ellipse,minimum width=2cm, minimum height=1cm,inner sep=0pt] at (0, 2.5) (R) {$\mathbbm{R}$};
    \def\dx{-3};
    \node[draw,fill=teal!50,fill opacity=1, text opacity=1.0,ellipse,minimum width=5cm, minimum height=3cm,inner sep=0pt] at (\dx, 0.15) (PN) {\hspace{-3.7cm}Polynomial};
    \node[draw,fill=brown!75,fill opacity=1, text opacity=1.0,ellipse,minimum width=3.5cm, minimum height=1.0cm,inner sep=0pt] at (\dx, 0.7) (P1) {\hspace{-1.0cm}Largest order};
    \node[draw,fill=brown!40,fill opacity=1, text opacity=1.0,ellipse,minimum width=3.5cm, minimum height=1.0cm,inner sep=0pt] at (\dx, -0.4) (P2) {\hspace{-0.8cm}Largest 2 orders};
    \node[draw,fill=brown,fill opacity=1, text opacity=1.0,ellipse,minimum width=0.8cm, minimum height=0.3cm,inner sep=0pt] at (\dx, 1.0) (T) {};
    \node[draw,fill=brown!50,fill opacity=1, text opacity=1.0,ellipse,minimum width=0.8cm, minimum height=0.3cm,inner sep=0pt] at (\dx, -0.1) (T2) {};
    \node at (\dx, -1.2) {$\ldots$};
    \node[above = 1cm] at (T) (textT) {Tropical};
    \node[below = 2cm, left=4cm] (textT2) {Extended Tropical};
    \draw[black,-latex] (textT) -- (T);
    \draw[black,-latex] (textT2) -- (T2);

    \node[draw,fill=yellow,fill opacity=0.5, text opacity=1.0,ellipse,minimum width=5cm, minimum height=3cm,inner sep=0pt] at (0, 0) (SN) {};
    \node[below of=1] at (SN) {Set};
    \node[draw,fill=red!70,fill opacity=0.5, text opacity=1.0,ellipse,minimum width=4cm, minimum height=1.5cm,inner sep=0pt] at (-0.5, 0) (S1) {\hspace{1.5cm}Bit string};
\end{tikzpicture}}

    \caption{The tensor network element types used in this work and their relations.
    The overlap between two ellipses indicates that a new algebra can be created by combining those two types of algebra. ``Largest order'' and ``Largest 2 orders'' mean truncating the polynomial by only keeping its largest or largest two orders.
    The purposes of these element types can be found in \Cref{tbl:generictypes}.}
     \label{fig:venn-diagram}
\end{figure}
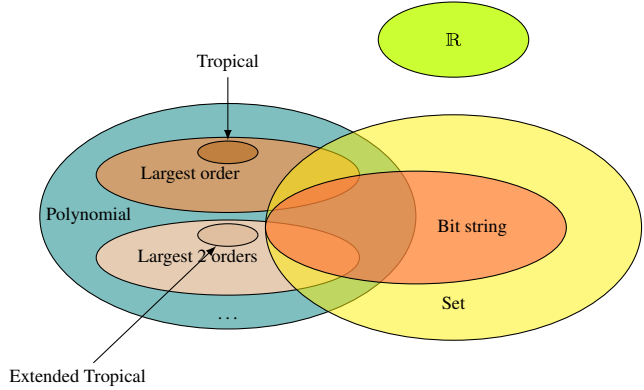

\begin{table}[t!]\centering
\begin{minipage}{\columnwidth}
\ra{1.3}
        \begin{tabularx}{\textwidth}{sb}\toprule
            \hline
   \textbf{Element type}     & \textbf{Solution space property} \\
   {$\mathbbm{R}$}     & {Counting of all independent sets} \\
   {Polynomial} (\Cref{eq:polynomial})     & {Independence polynomial} \\
   {Tropical (\Cref{eq:tropical})}    & {Maximum independent set size} \\
   {Extended tropical of order $k$ (\Cref{eq:ext-tropical})}    & {Largest $k$ independent set sizes} \\
   {Polynomial truncated to $k$-th order (\Cref{eq:countingtropical} and \Cref{eq:max2poly})}     & {$k$ largest independent sizes and their degeneracy} \\
   {Set} (\Cref{eq:set})     & {Enumeration of independent sets} \\
   {Sum-Product expression tree} (\Cref{eq:expr})     & {Sampling of independent sets} \\
   {Polynomial truncated to largest order combined with bit string} (\Cref{eq:singleconfig})     & {Maximum independent set size and one of such configurations} \\
   {Polynomial truncated to $k$-th order combined with set} (\Cref{eq:countingtropicalset})    & {$k$ largest independent set sizes and their enumeration} \\
            \bottomrule
        \end{tabularx}
    \caption{Tensor element types and the independent set properties that can be computed using them.}\label{tbl:generictypes}
\end{minipage}
\end{table}

\section{Tensor network representation of independent sets} \label{sec:tnmap}
This section describes the reduction of the independent set problem to a tensor network contraction problem.
An alternative interpretation, perhaps more accessible to physicists, can be found in \Cref{sec:energymodel}, where we introduce the reduction from the energy model of hardcore lattice gases~\cite{Dyre2016, Fernandes2007}.
Let $G=(V,E)$ be an undirected graph with each vertex $v\in V$ being associated with a integer weight $w_v$.
An independent set $I \subseteq V$ is a set of vertices that for any vertex pair $u,v \in I$, we have $(u, v) \not\in E$; we refer to this constraint as the independence constraint. 
The independent set problem on $G$ can be encoded as a tensor network $\mathcal{N}_{\rm IS}(G)$
\begin{equation}\label{eq:mistensornetwork}
\begin{split}
    \Lambda &= \{s_v \mid v \in V\},\\
    \mathcal{T} &= \{W^{(v)}_{s_v} \mid v\in V\} \cup \{B^{(u, v)}_{s_us_v} \mid (u, v) \in E\},\\
    \boldsymbol{\sigma}_o &= \varepsilon,
    \end{split}
\end{equation}
where for each vertex $v$, we define a parameterized rank-one tensor associated with it as
\begin{equation}
    W^{(v)} = \left(\begin{matrix}
        1 \\
        x_v^{w_v}
    \end{matrix}\right),\label{eq:vertextensor}
\end{equation}
and for each edge $(u, v) \in E$, we define a matrix $B$ as
\begin{equation}
    \qquad \quad 
       B^{(u,v)} = \left(\begin{matrix}
        1  & 1\\
        1 & 0
    \end{matrix}\right). \label{eq:edgetensor}
\end{equation}
We map each vertex $v\in V$ to a label $s_v \in \{0, 1\}$, where we use $0$ or $1$ to denote a vertex is absent or present in $I$, respectively.
These labels can be used as subscripts of tensors to index tensor elements, e.g.\ $W^{(v)}_0=1$ is the first element associated with $s_v=0$ and $W^{(v)}_1=x_v^{w_v}$ is the second element associated with $s_v=1$, where $x_v$ is an  element of some commutative semiring (e.g., listed in \Cref{tbl:generictypes}) associated with vertex $v$ and its power with an integer is defined by repeated multiplication.
The labels associated to the output tensor is an empty string $\varepsilon$, meaning this tensor has rank $0$, i.e.~the output is a scalar.
The independence constraint is encoded in the edge tensors, where we use $B^{(u, v)}_{11} = 0$ to denote two vertices connected by an edge $(u, v)$ cannot be both in the independent set. 
The contraction of this tensor network is
\begin{equation}\label{eq:idp}
    P(G) = \sum\limits_{s_1, s_2, \ldots, s_{|V|} = 0}^{1} \prod\limits_{v\in V} W^{(v)}_{s_v} \prod\limits_{(u,v) \in E} B^{(u,v)}_{s_u s_v},
\end{equation}
where the summation runs over all $2^{|V|}$ vertex configurations $\{s_1, s_{2}, \ldots,s_{|V|}\}$ and accumulates the product of tensor elements to the output $P$. 
A vertex tensor element $W^{(v)}_{s_v}$ contributes a multiplicative factor $x_v^{w_v}$ whenever $v$ is in the set.

\begin{example}\label{eg:twonode}
Here, we show a minimum example of mapping the independent problem of a 2-vertex complete graph K2 (left) to a tensor network (right).

\centerline{\begin{tikzpicture}[
dot/.style = {circle, fill, minimum size=#1,
            inner sep=0pt, outer sep=0pt},
dot/.default = 6pt  
                ]  
    \def\dx{0};
    \def\r{0.5cm}
    \def\wr{0.25cm}
    \node[dot=\r, fill=black] at (\dx,0) (a) {\color{white}{a}};
    \node[dot=\r, fill=black] at (2.5+\dx,0) (b) {\color{white}{b}};
    \draw [black,thick] (a) -- (b);

    \def\dx{5};
    \def\r{0.25cm}
    \foreach \x/\y/\e in {1.25/0/ab}
        \node[color=black,fill=black,dot=2.5*\r] at (\x+\dx,\y) (\e) {\scriptsize \color{white}{$B^{(a, b)}$}};
    \foreach \x/\y/\v in {0.3/0/a, 2.2/0/b}
        \node[color=black] at (\x+\dx,\y) (s\v) {$s_\v$};
    \foreach \x/\y/\v in {-0.5/0/a, 3.0/0.0/b}{
        \node[dot=\wr, color=black] at (\x+\dx,\y) (w\v) {};
        \node at (\x+\dx,\y+0.5) () {\scriptsize \color{black}{$W^{(v)}$}};
    }
    \draw [cyan,thick] (wa) -- (sa);
    \draw [cyan,thick] (sa) -- (ab);
    \draw [red,thick] (wb) -- (sb);
    \draw [red,thick] (sb) -- (ab);
\end{tikzpicture}}

In the graphical representation of the tensor network on the right panel,
we use a circle to represent a tensor, a cyan-colored hyperedge to represent the degree of freedom $s_a$,
and a hyperedge in red to represent the degree of freedom $s_b$.
Tensors sharing the same degree of freedom are connected by the same hyperedge.
The contraction of this tensor network has the following form:
\begin{equation}
    \begin{split}
    P({\rm K2}) &=
    \left(\begin{matrix} 1 & x_a^{w_a} \end{matrix}\right)
    \left(\begin{matrix} 1 & 1\\ 1  & 0 \end{matrix}\right)
    \left(\begin{matrix} 1 \\ x_b^{w_b} \end{matrix}\right)\\
    &= 1 + x_a^{w_a} + x_b^{w_b}.
    \end{split}
\end{equation}
The resulting polynomial represents 3 different independent sets $\{\}$, $\{a\}$, and $\{b\}$ with weights $0$, $w_a$, and $w_b$, respectively.
\end{example}
 
For a general graph, it is computationally inefficient to evaluate \Cref{eq:idp} by directly summing up the $2^{|V|}$ products.
A better approach to evaluate a tensor network is: to find a good pair-wise tensor contraction order as a binary tree and then contract two tensors at a time by this order.
\begin{theorem}\label{thm:complexreal}
    The tensor network in \Cref{eq:mistensornetwork} for the independent set problem on graph $G = (V,E)$ can be contracted in $\cc = O(|E|)2^{O({\rm tw}(G))}$ number of additions and multiplications.
\end{theorem}
\begin{proof}
Let us denote the line graph, a graph obtained by mapping an edge in the original graph to a vertex and connect two vertices if and only if their associated edges in the original graph share a common vertex, of the hypergraph representation of a tensor network $\mathcal{N}$ as $L(\mathcal{N})$.
A contraction order of $\mathcal{N}$ corresponds to a tree decomposition of $L(\mathcal{N})$ and the largest intermediate tensor has a rank equal to the width of its tree decomposition~\cite{Markov2008}.
Therefore, an optimal (in terms of space complexity) contraction order corresponds to the tree decomposition of $L(\mathcal{N})$ with the smallest width (or the treewidth ${\rm tw}(L(\mathcal{N})$). The contraction complexity is $O(|\mathcal{N}|)2^{O({\rm tw}(L(\mathcal{N})))}$, where $|\mathcal{N}|$ is the number of tensors in $\mathcal{N}$~\cite{Markov2008}.
For the independent set problem on graph $G = (V, E)$, the line graph of the hypergraph representation of the tensor network in \Cref{eq:mistensornetwork} is isomorphic to the graph $G$ up to some isolated vertices, hence the contraction complexity of this tensor network with an optimal contraction order is $O(|E|)2^{O({\rm tw}(G))}$.
\end{proof}
 
In practice, it is difficult to find an optimal contraction order for large tensor networks because finding the treewidth is a well-known NP-hard problem.
However, it is easy to find a close-to-optimal contraction order within typically a few minutes using a heuristic algorithm~\cite{Kourtis2019, Kalachev2021}.
For large-scale applications, it is also possible to slice over certain degrees of freedom to reduce the space complexity, i.e.\
loop over possible combinations of certain degrees of freedom so that one can have a smaller tensor network inside the loop since these degrees of freedoms are fixed.

\begin{example}\label{eg:tensorcontraction}
In this example, we map a 5-vertex graph (left) to a tensor network (right) and show how optimizing the contraction order reduces the time and space complexities.
    
\centerline{\begin{tikzpicture}[
dot/.style = {circle, fill, minimum size=#1,
            inner sep=0pt, outer sep=0pt},
dot/.default = 6pt  
                ]  
    \def\dx{0};
    \def\r{0.25cm}
    \filldraw[fill=black] (\dx,0) circle [radius=\r];
    \filldraw[fill=black] (\dx,1.5) circle [radius=\r];
    \filldraw[fill=black] (1.5+\dx,0) circle [radius=\r];
    \filldraw[fill=black] (1.5+\dx,1.5) circle [radius=\r];
    \filldraw[fill=black] (2.5+\dx,2.5) circle [radius=\r];
    \draw [black,thick] (\dx,0) -- (\dx,1.5);
    \draw [black,thick] (\dx,0) -- (1.5+\dx,0);
    \draw [black,thick] (\dx,1.5) -- (1.5+\dx,1.5);
    \draw [black,thick] (1.5+\dx,0) -- (1.5+\dx,1.5);
    \draw [black,thick] (1.5+\dx,0) -- (\dx,1.5);
    \draw [black,thick] (2.5+\dx,2.5) -- (1.5+\dx,1.5);
    \node[color=white] at (\dx,0) {a};
    \node[color=white] at (\dx,1.5) {b};
    \node[color=white] at (1.5+\dx,0) {c};
    \node[color=white] at (1.5+\dx,1.5) {d};
    \node[color=white] at (2.5+\dx,2.5) {e};
    \def\dx{5};
    \def\r{0.25cm}
    \foreach \x/\y/\u/\v in {0.75/0/a/c, 0/0.75/a/b, 1.5/0.75/c/d, 0.75/1.5/b/d, 0.75/0.75/b/c, 2/2/d/e}
        \node[color=white,fill=black,dot=2.5*\r] at (\x+\dx,\y) (\u\v) {\scriptsize $B^{(\u,\v)}$};
    \foreach \x/\y/\v in {0/0/a, 0/1.5/b, 1.5/0/c, 1.5/1.5/d, 2.5/2.5/e}
        \node[color=black] at (\x+\dx,\y) (\v) {$s_\v$};
    \foreach \x/\y/\v in {-0.5/-0.5/a, -0.5/2.0/b, 2.0/-0.5/c, 2.0/1.0/d, 3.0/2.0/e}
        \node[color=white,fill=black,dot=\r] at (\x+\dx,\y) (\v\v) {};
    \foreach \x/\y/\v in {-0.5/-0.5/a, -0.5/2.0/b, 2.0/-0.5/c, 2.0/1.0/d, 3.0/2.0/e}
        \node[color=black] at (\x+\dx+0.6,\y) {\scriptsize $W^{(\v)}$};
    \draw [cyan,thick] (a) -- (aa);
    \draw [cyan,thick] (a) -- (ab);
    \draw [cyan,thick] (a) -- (ac);
    \draw [blue,thick] (b) -- (bb);
    \draw [blue,thick] (b) -- (ab);
    \draw [blue,thick] (b) -- (bc);
    \draw [blue,thick] (b) -- (bd);
    \draw [red,thick] (c) -- (cc);
    \draw [red,thick] (c) -- (ac);
    \draw [red,thick] (c) -- (bc);
    \draw [red,thick] (c) -- (cd);
    \draw [green,thick] (d) -- (dd);
    \draw [green,thick] (d) -- (bd);
    \draw [green,thick] (d) -- (de);
    \draw [green,thick] (d) -- (cd);
    \draw [orange,thick] (e) -- (ee);
    \draw [orange,thick] (e) -- (de);
\end{tikzpicture}}

One can represent a possible pair-wise contraction of tensors as a binary tree structure:

\centerline{\begin{tikzpicture}[]
    \def\dx{0};
    \def\dy{0};
    \def\d{1};
    \def\a{1.0};
    \def\b{1};
    \def\ya{-1.0};
    \def\yb{-0.8};
    \node[] at (\dx+5*\a, \dy+5*\b) (R1) {$R[n^2]$};
    \node[] at (\dx+3*\a, \dy+4*\b) (S1) {$S_{bc}[n^3]$};
    \node[] at (\dx+2*\a, \dy+3*\b) (Q1) {$Q_{bd}[n^2]$};
    \node[] at (\dx+8.5*\a, \dy+3*\b) (Q2) {$Q_{bc}[n^2]$};
    \node[] at (\dx+\a, \dy+2*\b) (P1) {$P_{d}[n]$};
    \node[] at (\dx+9*\a, \dy+2*\b) (P2) {$P_{bc}[n^3]$};
    \node[] at (\dx+0.5*\d, \dy+\b) (M1) {$M_{d}[n^2]$};
    \node[] at (\dx+3.5*\d, \dy+\b) (M2) {$M_{bd}[n^2]$};
    \node[] at (\dx+5.5*\d, \dy+\b) (M3) {$M_{cd}[n^2]$};
    \node[] at (\dx+9.5*\d, \dy+\b) (M4) {$M_{ac}[n^2]$};
    \node[] at (\dx, \dy) (Bde) {$B_{de}$};
    \node[] at (\dx+\d, \dy) (We) {$W_e$};
    \node[] at (\dx+2*\d, \dy) (Wd) {$W_d$};
    \node[] at (\dx+3*\d, \dy) (Bbd) {$B_{bd}$};
    \node[] at (\dx+4*\d, \dy) (Wb) {$W_b$};
    \node[] at (\dx+5*\d, \dy) (Bcd) {$B_{cd}$};
    \node[] at (\dx+6*\d, \dy) (Wc) {$W_c$};
    \node[] at (\dx+7*\d, \dy) (Bbc) {$B_{bc}$};
    \node[] at (\dx+8*\d, \dy) (Bab) {$B_{ab}$};
    \node[] at (\dx+9*\d, \dy) (Bac) {$B_{ac}$};
    \node[] at (\dx+10*\d, \dy) (Wa) {$W_a$};
    \draw[] (Bde) -- (M1);
    \draw[] (We) -- (M1);
    \draw[] (Bbd) -- (M2);
    \draw[] (Wb) -- (M2);
    \draw[] (Bcd) -- (M3);
    \draw[] (Wc) -- (M3);
    \draw[] (Bac) -- (M4);
    \draw[] (Wa) -- (M4);
    \draw[] (M1) -- (P1);
    \draw[] (Wd) -- (P1);
    \draw[] (Bab) -- (P2);
    \draw[] (M4) -- (P2);
    \draw[] (P1) -- (Q1);
    \draw[] (M2) -- (Q1);
    \draw[] (Bbc) -- (Q2);
    \draw[] (P2) -- (Q2);
    \draw[] (Bbc) -- (Q2);
    \draw[] (P2) -- (Q2);
    \draw[] (Q1) -- (S1);
    \draw[] (M3) -- (S1);
    \draw[] (S1) -- (R1);
    \draw[] (Q2) -- (R1);
\end{tikzpicture}}
The contraction process goes from bottom to top, where the root node stores the contraction result, the leaves are input tensors, and the rest of the nodes are all intermediate contraction results.
Tensor subscripts are indices so that the number of subscripts indicates the space complexity to store this tensor.
The contraction complexity to generate a tensor is annotated in the square brackets, where $n$ is the dimension of the degree of freedoms, which is $2$ in a tensor network mapped from an independent set problem.
One can easily check the largest tensor in contraction has a space complexity $O(n^2)$ and this is the smallest among all possible contraction trees, i.e.\ the treewidth of the original $5$-vertex graph is $2$. The time complexity is $O(n^3)$, which is much smaller than that of direct evaluation $O(n^5)$.
\end{example}

\section{Independence polynomial}\label{sec:indpoly}

Let $x_i = x$ and $w_i = 1$, \Cref{eq:idp} corresponds to the independence polynomial:
\begin{equation}\label{eq:idpdef}
I(G, x) = \sum_{k=0}^{\alpha(G)} a_k x^k,
\end{equation}
where $a_k$ is the number of independent sets of size $k$, $\alpha(G) \equiv \max_{I}\sum_{v\in I}w_v$ is the size of a maximum independent set and it is also called the independence number.
An independence polynomial is a useful graph characteristic related to, for example, the partition functions~\cite{Lee1952,Yang1952} and Euler characteristics of the independence complex~\cite{Bousquet2008, Levit2009}.
By assigning a real number to $x$, one can evaluate this independence polynomial for this specific value directly using tensor network contraction.
For example, the total number of independent sets can be evaluated as $I(G, 1)$.
However, instead of evaluating this polynomial for a certain value, we are more interested in knowing the coefficients of this polynomial, because this quantity tells us the counting of independent sets at different sizes.
To this end, let us create a polynomial number data type by representing a polynomial $a_0 + a_1 x + \ldots + a_k x^k$ as a coefficient vector $a = (a_0, a_1, \ldots, a_k) \in \mathbb{R}^k$, e.g. $x$ is represented as $(0, 1)$.
Then we can define an algebra among coefficient vectors, including a redefinition of additive identity and multiplicative identity.
To avoid potential confusion, let us denote the additive identity as $\mymathbb{0}$, and the multiplicative identity is $\mymathbb{1}$.
The algebra between the polynomials number $a$ of order $k_a$ and $b$ of order $k_b$ is specified as
\begin{equation}
    \eqname{PN}
    \begin{split}
    a \oplus b &= (a_0 + b_0, a_1 + b_1, \ldots, a_{\max(k_a, k_b)} + b_{\max(k_a, k_b)}),\\
    a \odot b &= (a_0 + b_0, a_1b_0 + a_0b_1, a_{2}b_{0} + a_{1}b_{1} + a_{0}b_{2},  \ldots, a_{k_a} b_{k_b}),\\
    \mymathbb{0} &= (),  \\
    \mymathbb{1} &= (1), \label{eq:polynomial}
    \end{split}
\end{equation}
where $\oplus$ and $\odot$ are the standard polynomial addition and multiplication operations.

\begin{theorem}\label{thm:complexpoly}
    The independence polynomial~\cite{Harvey2018,Ferrin2014} of a graph $G = (V, E)$ can be computed in time $O(|V|\log(|V|))\cc$.
\end{theorem}
\begin{proof}
By doing the following replacement of tensor elements from the standard number type to the polynomial number type
\begin{equation}
    \begin{cases}
    1 &\rightarrow \mymathbb{1},\\
    0 &\rightarrow \mymathbb{0},\\
    x_v^{w_v} &\rightarrow (0, 1),
    \end{cases}
\end{equation}
the tensors $W$ and $B$, which are introduced in the previous section (\Cref{sec:tnmap}), can thus be written as 
\begin{equation}
    \left(W^{\rm PN}\right)^{(v)} = \left(\begin{matrix}
        \mymathbb{1} \\
        (0,1)
    \end{matrix}\right),   
    \qquad \qquad
        \left(B^{\rm PN}\right)^{(u, v)} = \left(\begin{matrix}
        \mymathbb{1}  & \mymathbb{1} \\
        \mymathbb{1} & \mymathbb{0}
    \end{matrix}\right).
\end{equation}
By contracting the tensor network with this polynomial type, we have the exact representation of the independence polynomial.
In \Cref{eq:polynomial}, the addition can be computed in time $O(k)$, where $k=\max(k_a, k_b)$, and the multiplication can be evaluated in time $O(k\log(k))$ using the convolution theorem~\cite{Schonhage1971}.
Since $k$ is upper bounded by the maximum independent set size $\alpha(G) \leq |V|$, the time complexity of element-wise addition and multiplication operation is upper bounded by $|V|\log(|V|)$. Combining with \Cref{thm:complexreal}, the overall time complexity is $O(|V|\log(|V|))\cc$.
\end{proof}

In practice, using the polynomial type suffers a space overhead proportional to $\alpha(G)$ because each polynomial requires a vector of such size to store the coefficients. 
One may argue that one can first evaluate this polynomial at different $x$ being a real number,
and then apply the Gaussian elimination procedure to fit the coefficients of this polynomial.
However, in practice this seemingly more time and space-efficient approach suffers from precision issues.
The data ranges of standard integer types are too small to cover many practical use cases,
while the floating-point numbers may have round-off errors that are much larger than the value itself.
These are because the number of independent sets at different sizes may vary by tens or even hundreds of orders of magnitude.
For practical methods to evaluate these coefficients, we refer readers to \Cref{sec:finitefield}, where we provide an accurate and memory-efficient method to find the polynomial by contracting and fitting on finite field algebra.
For simplicity, we use this less efficient polynomial algebra for the discussion in the main text.

\section{Maximum independent sets and its counting}\label{sec:counting}
\subsection{The number of independent sets}
\begin{theorem}\label{thm:complexrealnum}
    Let $G = (V, E)$ be a graph. The total number independent sets of $G$ can be computed in time $\cc$.
\end{theorem}
\begin{proof}
Let $x = 1$ be a real constant number; the independence polynomial in the previous section becomes
\begin{equation}
I(G, 1) = \sum_{k=0}^{\alpha(G)}a_k,
\end{equation}
which corresponds to the total number of independent sets.
Since the real number addition and multiplication can be computed in $O(1)$ time, thus
proving the theorem by \Cref{thm:complexreal}.
\end{proof}

\subsection{Tropical algebra for finding the MIS size and counting MISs}
Let $x=\infty$, the independence polynomial in the previous section becomes
\begin{equation}
I(G, \infty) = \lim_{x\rightarrow \infty}\sum_{k=0}^{\alpha(G)}a_k x^k = a_{\alpha(G)} \infty^{\alpha(G)},
\end{equation}
where all terms except the one with the largest order vanish. We can thus replace the polynomial type $a = (a_0, a_1, \ldots, a_k)$ with a new type that has two fields: the largest exponent $k$ and its coefficient $a_k$.
From this, we can define a new algebra as
\begin{equation}
    \eqname{P1}
\begin{split}
    a_x\infty^x \oplus a_y\infty^y &= \begin{cases}
        (a_x + a_y)\infty^{\max(x,y)}, & x = y\\
        a_y\infty^{\max(x,y)}, & x < y\\
        a_x\infty^{\max(x,y)}, & x > y
    \end{cases}, \\
    a_x\infty^x \odot a_y\infty^y &= a_x a_y\infty^{x+y}\\
    \mymathbb{0} &= 0\infty^{-\infty}\\
    \mymathbb{1} &= 1\infty^{0}.
\end{split}
\label{eq:countingtropical}
\end{equation}
Here, we have generalized the previous polynomial to the Laurent polynomial to define the zero-element properly.
To implement this algebra programmatically, we create a data type with two fields $(x, a_x)$ to store the MIS size and its counting,
and define the above operations and constants correspondingly.
If one is only interested in finding the MIS size, one can drop the counting field.
The algebra of the exponents becomes the max-plus tropical algebra~\cite{Maclagan2015, Moore2011}:
\begin{equation}\eqname{T}
    \begin{split}
        x \oplus y &= \max(x,y)\\
        x \odot y &= x + y\\
        \mymathbb{0} &= -\infty\\
        \mymathbb{1} &= 0.
    \end{split}\label{eq:tropical}
\end{equation}
Algebra \Cref{eq:tropical} and \Cref{eq:countingtropical} are the same as those used in Liu et al.~\cite{Liu2021} to compute the spin glass ground state energy and its degeneracy.

\begin{theorem}\label{thm:complextropical}
    Let $G = (V, E)$ be a graph. Its maximum independent set size $\alpha(G)$ can be computed in time $\cc$.
\end{theorem}
\begin{proof}
By replacing the tensor elements from the standard number type to the tropical numbers in \Cref{eq:tropical}, the vertex and edge tensors transforms to
\begin{equation}\label{eq:tensor-topical}
    \left(W^{\rm T}\right)^{(v)} = \left(\begin{matrix}
        \mymathbb{1} \\
        \infty^{w_v}
    \end{matrix}\right),   
    \qquad \qquad
        \left(B^{\rm T}\right)^{(u, v)} = \left(\begin{matrix}
        \mymathbb{1}  & \mymathbb{1} \\
        \mymathbb{1} & \mymathbb{0}
    \end{matrix}\right).
\end{equation}
The maximum independent set can be obtained by contracting a tensor network with the above vertex tensors and edge tensors.
Since the element-wise addition and multiplication can be computed in $O(1)$ time, the complexity of contracting this tensor network is $\cc$ by \Cref{thm:complexreal}.
\end{proof}
\subsection{Truncated polynomial algebra for counting independent sets of large size}
Instead of counting just the MISs, one may also be interested in counting the independent sets with the largest several sizes.
For example, if one is interested in counting only $a_{\alpha(G)}$ and $a_{\alpha(G)-1}$, we can define a truncated polynomial algebra by keeping only the largest two coefficients in the polynomial in \Cref{eq:polynomial} as:
\begin{equation}
    \eqname{P2}
    \begin{split}
    a \oplus b &= (a_{\max(k_a, k_b)-1} + b_{\max(k_a, k_b)-1}, a_{\max(k_a, k_b)} + b_{\max(k_a, k_b)}),\\
    a \odot b &= (a_{k_a-1} b_{k_b}+a_{k_a} b_{k_b-1}, a_{k_a} b_{k_b}),\\
    \mymathbb{0} &= (), \\
    \mymathbb{1} &= (1).\label{eq:max2poly}
    \end{split}
\end{equation}
In the program, we thus need a data structure that contains three fields, the largest order $k$, and the coefficients for the two largest orders $a_k$ and $a_{k-1}$.
This approach can clearly be extended to calculate more independence polynomial coefficients and is more efficient than calculating the entire independence polynomial.
Similarly, one can also truncate the polynomial and keep only its smallest several orders.
It can be used, for example, to count the maximal independent sets with the smallest cardinality, where a maximal independent set is an independent set that cannot be made larger by adding a new vertex into it without violating the independence constraint.
As will be shown below, this algebra can also be extended to enumerate those large-size independent sets.

\begin{theorem}
    The number of the largest $K$ independent sets of a graph $G = (V, E)$ can be computed in time $O(K\log K)\cc$.
\end{theorem}
The proof is similar to the proof of \Cref{thm:complexpoly}, except only $K$ largest coefficients are involved in the addition and multiplication.

\section{Enumerating and sampling independent sets}\label{sec:enumeration}

\subsection{Set algebra for configuration enumeration}

The configuration enumeration of independent sets include, for example, the enumeration of all independent sets, the enumeration of all MISs, and the enumeration of independent sets with the largest several sizes.
Recall that in the definition of a vertex tensor in \Cref{eq:vertextensor}, variables carry labels, so that one can read out all independent sets directly from the output polynomials.
The multiplication between labelled variables is commutative while the summation of labelled variables forms a set.
Intuitively, one can use a bit string as the representation of a labelled variable and use the bit-wise or $\lor^\circ$ as the multiplication operation.
For example, in a 5-vertex graph, $x_2$ and $x_5$ can be represented as $01000$ and $00001$ respectively and their multiplication $x_2x_5$ can be represented as $01000 \lor^\circ 00001 = 01001$.
To enumerate all independent sets, we designed an algebra on sets of bit strings:
\begin{equation}
\eqname{SN}
\begin{split}
    s \oplus t &= s \cup t\\
    s \odot t &= \{\sigma \lor^\circ \tau \, \mid \, \sigma \in s, \tau \in t\}\\
    \mymathbb{0} &= \{\}\\
    \mymathbb{1} &= \{0^{\otimes |V|}\},
\end{split}\label{eq:set}
\end{equation}
where $s$ and $t$ are each a set of $|V|$-bit strings.
\begin{example}\label{eg:setalgebra}
    For elements that are bit strings of length $5$, we have the following set algebra
\begin{equation*}
\begin{split}
    &\{00001\} \oplus \{01110, 01000\} = \{01110, 01000\} \oplus \{00001\} = \{00001,01110, 01000\}\\
    &\{00001\} \oplus \{\} = \{00001\}\\
&\\
    &\{00001\} \odot \{01110, 01000\} = \{01110, 01000\} \odot \{00001\} = \{01111, 01001\}\\
    &\{00001\} \odot \{\} = \{\}\\
    &\{00001\} \odot \{00000\} = \{00001\}.
\end{split}
\end{equation*}
\end{example}

\begin{lemma}\label{thm:complexset}
    The independent sets of a graph $G = (V, E)$ can be enumerated in time $O\left(|I|*|V| \right)\cc$, where $|I|$ is the number of independent sets.
\end{lemma}
\begin{proof}
To enumerate over $I$,  we initialize the variable $x_{i}$ in the vertex tensor to $x_i = \{\boldsymbol{e}_{i}\}$, where $\boldsymbol{e}_i$ is a basis bit string of size $|V|$ that has only one non-zero value at location $i$.
The vertex and edge tensors are thus
\begin{equation}
    \left(W^{\rm SN}\right)^{(v)} = \left(\begin{matrix}
        \mymathbb{1} \\
        \{\boldsymbol{e}_{i}\}
    \end{matrix}\right),   
    \qquad \qquad
        \left(B^{\rm SN}\right)^{(u, v)} = \left(\begin{matrix}
        \mymathbb{1}  & \mymathbb{1} \\
        \mymathbb{1} & \mymathbb{0}
    \end{matrix}\right).
\end{equation}
The contraction of this tensor network is a set with its elements being all independent set configurations. The time complexity of element-wise addition and multiplication is upper bounded by $|I|*|V|$, where $|V|$ comes from the linear cost of representing a $|V|$-vertex configuration. Combining with \Cref{thm:complexreal}, the overall time complexity is $O(|I|*|V|)\cc$.
In practice, the huge multiplicative factor $|I|$ only appear in the last step of tensor contraction.
This complexity only serves as an upper bound for the actual performance of the algorithm.
\end{proof}

This set algebra can serve as the coefficients in \Cref{eq:polynomial} to enumerate independent sets of all different sizes, \Cref{eq:countingtropical} to enumerate all MISs, or \Cref{eq:max2poly} to enumerate all independent sets of sizes $\alpha(G)$ and $\alpha(G)-1$.
As long as the coefficients in a truncated polynomial are members of a commutative semiring, the polynomial itself is a commutative semiring.

For example, to enumerate only the MISs, we can define a combined element type $s_{k}\infty^k$, where the coefficients follow the algebra in \Cref{eq:set} and the exponents follow the max-plus tropical algebra.
The combined operations become: 
\begin{equation}
\eqname{P1+SN}
\begin{split}
    s_x\infty^x \oplus s_y\infty^y &= \begin{cases}
        (s_x \cup s_y)\infty^{\max(x,y)}, & x = y\\
        s_y\infty^{\max(x,y)}, & x < y\\
        s_x\infty^{\max(x,y)}, & x > y
    \end{cases},\\
    s_x\infty^x \odot s_y\infty^y &= \{\sigma \lor^\circ \tau \mid \sigma \in s_x, \tau \in s_y\}\infty^{x+y},\\
    \mymathbb{0} &= \{\}\infty^{-\infty},\\
    \mymathbb{1} &= \{0^{\otimes |V|}\}\infty^{0}. \label{eq:countingtropicalset}
\end{split}
\end{equation}

\begin{lemma}\label{thm:complexmis}
    Let $G = (V, E)$ be a graph and $\mathcal{I}_\alpha$ be the set of all maximum independent sets of $G$.
    $\mathcal{I}_\alpha$ can be computed in time $O(|\mathcal{I}_\alpha|*|V|)\cc$.
\end{lemma}
\begin{proof}
By replacing the tensor elements in \Cref{sec:tnmap} with the number type defined in \Cref{eq:countingtropicalset} and with variable $x_v^{w_v} = \{\boldsymbol{e}_{v}\}\infty^{w_v}$, the vertex tensor and edge tensor become
\begin{equation}
    \left(W^{\rm P1+SN}\right)^{(v)} = \left(\begin{matrix}
        \mymathbb{1} \\
        \{\boldsymbol{e}_{v}\}\infty^{w_v}
    \end{matrix}\right),   
    \qquad \qquad
        \left(B^{\rm P1+SN}\right)^{(u, v)} = \left(\begin{matrix}
        \mymathbb{1}  & \mymathbb{1} \\
        \mymathbb{1} & \mymathbb{0}
    \end{matrix}\right).
\end{equation}
The enumeration of all MIS configurations corresponds to the contraction of this tensor network. The time complexity of element-wise addition and multiplication is upper bounded by the number of maximum independent sets $|\mathcal{I}_\alpha|$. Combining with \Cref{thm:complexreal}, the overall time complexity is $O(|\mathcal{I}_\alpha|*|V|)\cc$.
\end{proof}
However, direct contraction might have significant space overheads for keeping too many intermediate states irrelevant to the final maximum independent sets.
We introduce the bounding technique in \Cref{sec:bounding} to avoid this issue.
One may also be interested in the more studied maximal independent sets~\cite{Bron1973, Eppstein2010, Johnson1988} enumeration.
We discuss this in \Cref{sec:maximal} since it requires using a different tensor network structure.

\begin{lemma}\label{thm:complexsetk}
    Let $G = (V, E)$ be a graph and $\mathcal{I}_k$ be the set of independent sets with size $k$.
    The independent sets with size $k > \alpha(G)-K$ of $G$ can be computed in time $O\left(|V|* \sum_{k=\alpha(G)-K+1}^{\alpha(G)}|\mathcal{I}_k| \right)\cc$.
\end{lemma}
\begin{proof}
This can be proved by combining the set algebra \Cref{eq:set} with the polynomial number truncated to largest $K$ orders (\Cref{eq:max2poly}).
The number of bitstring operations in the addition and multiplication of the joint algebra is upper bounded by the elements in the sets $\sum_{k=\alpha(G)-K+1}^{\alpha(G)}|\mathcal{I}_k|$.
Combining with \Cref{thm:complexreal}, the overall time complexity of contracting this tensor network is $O\left(|V|* \sum_{k=\alpha(G)-K+1}^{\alpha(G)}|\mathcal{I}_k|\right)\cc$.
\end{proof}

If one is interested in obtaining only one MIS configuration, they can just keep one configuration in each tensor element to save the computational effort.
By replacing the sets of bit strings in \Cref{eq:set} with a single bit string, we have the following algebra
\begin{equation}
\eqname{S1}
\begin{split}
    \sigma \oplus \tau &= \texttt{select}(\sigma, \tau), \\
    \sigma \odot \tau &= (\sigma\lor^\circ \tau),\\
    \mymathbb{0} &= 1^{\otimes |V|}, \\
    \mymathbb{1} &= 0^{\otimes |V|}.
\end{split}\label{eq:singleconfig}
\end{equation}
The \texttt{select} function picks one of $\sigma$ and $\tau$ by some criteria.
It can be picking the one smaller in the lexicographical order such that the addition operation is commutative and associative.
In most cases, it is completely fine for the \texttt{select} function to pick a random one (not commutative and associative anymore) to generate a random MIS.

\begin{proposition}
    Let $G = (V, E)$ be a graph. One of its maximum independent set can be computed in time $O(|V|)\cc$.
\end{proposition}
The proof is similar to \Cref{thm:complexmis}, except the set algebra is replace by the one in \Cref{eq:singleconfig}, and the time complexity of element-wise addition and multiplication no longer depends on $|\mathcal{I}_\alpha|$.
The linear dependency of $|V|$ in the time complexity can be remove using the back propagation technique, at the cost of additional space.
\begin{theorem}
    Let $G = (V, E)$ be a graph. One of its maximum independent set can be computed in time $\cc$.
\end{theorem}
\begin{proof}
    Let $\alpha(G) = \max\limits_{s \in I} \sum_{v} s_v$ be the maximum independent set size. The optimal configuration can be obtained by differentiating over vertex tensors as
    \begin{equation}
    \begin{cases}
        1, &\frac{\partial \alpha(G)}{\partial s_v} = 1\\
        0, &\frac{\partial \alpha(G)}{\partial s_v} = 0
    \end{cases}.
    \end{equation}
    It can be numerically computed by differentiating over the process of the tropical tensor network contraction, where the backward rules can be found in \Cref{sec:bounding}. The time overhead of computing a single optimal solution is constant compared to only computing the contraction. Thus, the computing time complexity should be the same as in \Cref{thm:complextropical}.
\end{proof}

\subsection{Sampling from the extremely large configuration space}
When the problem size becomes larger, a set of all bitstrings might be impossible to fit into any type of storage.
To get something meaningful out of the configuration space, we use a binary sum-product expression tree as a compact representation of a set of configurations, i.e.\ instead of directly computing a set using the algebra in \Cref{eq:set}, we store the process of computing it.
Each node in this tree is a quadruple $(type, data, left, right)$, where $type$ is one of \texttt{LEAF}, \texttt{ZERO}, \texttt{SUM} and \texttt{PROD}, $data$ is a bit string as the content in a \texttt{LEAF} node, and $left$ and $right$ are left and right operands of a \texttt{SUM} or \texttt{PROD} node.

\begin{equation}
\eqname{EXPR}
\begin{split}
    s \oplus t &= (\texttt{SUM}, /, s, t)\\
    s \odot t &= (\texttt{PROD}, /, s, t)\\
    \mymathbb{0} &= (\texttt{ZERO}, /, /, /)\\
    \mymathbb{1} &= (\texttt{LEAF}, 0^{\otimes |V|}, /, /).
\end{split}\label{eq:expr}
\end{equation}

This algebra is a commutative semiring because we define the equivalence of two sum-product expression trees by comparing their expanded (using \Cref{eq:set}) forms rather than their storage.
Except using the sum-product expression tree directly as tensor elements, one can also let it be the coefficients of a (truncated) polynomial to compute such trees for independent sets with the largest several sizes.

\begin{theorem}\label{thm:complexexpr}
        The independent sets $I$ of a graph $G=(V,E)$ can be enumerated in time $\cc + O(|V|*|I|)$.
\end{theorem}
\begin{proof}
By replacing the tensor elements with the number type in \Cref{eq:expr} with $x_v^{w_v}$ mapped to $(\texttt{LEAF},\boldsymbol{e}_i,/,/)$, the vertex tensor and edge tensor become
\begin{equation}
    \left(W^{\rm EXPR}\right)^{(v)} = \left(\begin{matrix}
        \mymathbb{1} \\
        (\texttt{LEAF},\boldsymbol{e}_i,/,/)
    \end{matrix}\right),   
    \qquad
        \left(B^{\rm EXPR}\right)^{(u, v)} = \left(\begin{matrix}
        \mymathbb{1}  & \mymathbb{1} \\
        \mymathbb{1} & \mymathbb{0}
    \end{matrix}\right).
\end{equation}
The contraction result of this tensor network corresponds to a sum-product expression tree for the set of independent sets.
Unlike in \Cref{thm:complexset}, the contraction complexity is independent of the set size, i.e.~the time complexity is $\cc$ by \Cref{thm:complexreal}.
The time complexity to evaluating this expression tree is $O(|V|*|I|)$, where $|V|$ comes from the number of bits to represent a configuration.
Proving the theorem by adding two computing times.
\end{proof}

Similarly, we have the following corollary by combining the sum-product expression tree algebra with the truncated polynomial.

\begin{corollary}\label{thm:complexexprk}
        Let $G = (V, E)$ be a graph and $\mathcal{I}_k$ be the number of independent sets with size $k$. The independent sets with size $k > \alpha(G)-K$ can be computed in time $O(K\log K)\cc + O\left(|V|*\sum_{k=\alpha(G)-K+1}^{\alpha(G)}|\mathcal{I}_k| \right)$.
\end{corollary}
Again, the theorem can be proven by relating it with its set algebra version in \Cref{thm:complexsetk}. Because it is unlikely that one can collect all configurations represented by a sum-product expression tree into a set due to its space complexity, one can also use this construction to produce unbiased samples of the sum-product tree.

\begin{theorem}
    Let $G = (V, E)$ be a graph. Its independent sets with size $k > \alpha(G)-K$ can be unbiasedly sampled in time $\cc + O(|V|*|E| *M)$, where $M$ is the number of samples.
\end{theorem}
\begin{proof}
Similar to the proof of \Cref{thm:complexexpr}, we first obtain a sum-product tree representation of the configurations in time $\cc$.
Then we generate a sample using the following procedure with cost $O(|V|*|E|)$.
The sampling program starts from the root node and descends this tree recursively to the left and right siblings.
If a node has type \texttt{SUM}, the program draws samples from the left and right siblings with a probability decided by the size of each sub-tree and returns the union of samples.
Otherwise, if a node has type \texttt{PROD}, the program draws two sets of samples of equal sizes from its left and right siblings and returns the element-wise multiplication ($\lor^\circ$) of them.
The recursion stops at a \texttt{LEAF} node having size $1$ or a \texttt{ZERO} node having size $0$.
In a sum-product expression tree, the number of configurations of sub-trees can be determined easily as we will show in the following example.
\begin{example}
Let us consider the following sum-product expression tree 
\begin{equation*}
    (\texttt{SUM}, /, (\texttt{PROD}, /, A, B), (\texttt{SUM}, /, C, D)),
\end{equation*}
where sub-trees $A, B, C$ and $D$ can be any of the four types of nodes.
This sum-product expression tree can be represented diagrammatically as the following:

\centerline{\begin{tikzpicture}[]
    \def\dx{0};
    \def\a{1.0};
    \def\b{0.5};
    \def\ya{-1.0};
    \def\yb{-0.8};
    \node[] at (\dx, 0) (P1) {$\oplus$};
    \node[] at (\dx+\a, \ya) (P2) {$\oplus$};
    \node[] at (\dx-\a, \ya) (M1) {$\odot$};
    \node[] at (\dx-\a-\b, \ya+\yb) (A) {$A$};
    \node[] at (\dx-\a+\b, \ya+\yb) (B) {$B$};
    \node[] at (\dx+\a+\b, \ya+\yb) (C) {$C$};
    \node[] at (\dx+\a-\b, \ya+\yb) (D) {$D$};
    \draw[] (P1) -- (M1);
    \draw[] (P1) -- (P2);
    \draw[] (P2) -- (C);
    \draw[] (P2) -- (D);
    \draw[] (M1) -- (A);
    \draw[] (M1) -- (B);
\end{tikzpicture}}

The left and right siblings of the root node have sizes $|A| * |B|$ and $|C|+|D|$ respectively, while the root node size can be computed as $|A|* |B| + |C|+|D|$.
The sizes of $A, B, C$, and $D$ can be computed recursively until the program meets either a \texttt{LEAF} node or a \texttt{ZERO} node, which has a known size of $1$ or $0$.
\end{example}

Since the depth of the expression tree is $O(|E|)$ and the cost of each arithmetic operation is $O(|V|)$, the overall time complexity to generate a sample is $O(|V|*|E|)$, hence proving the theorem.
\end{proof}

Similarly, if one is only interested in obtaining independent sets with largest $K$ sizes, we have the following corollary.
\begin{corollary}
    Let $G = (V, E)$ be a graph. Its independent sets with size $k > \alpha(G)-K$ can be unbiasedly sampled in time $O(K\log K)\cc + O(|V|*|E|*M)$, where $M$ is the number of samples.
\end{corollary}

\section{Weighted graphs}\label{sec:weighted}

All the solution space properties and the corresponding algebra on unweighted graphs still hold for integer-weighted graphs, while for general weighted graphs, the independence polynomial is not well defined anymore.
For general weighted graphs, it is more useful to know the $k$ maximum weighted sets and their sizes.
They can be computed by the extended tropical algebra, which is a natural generalization of the max-plus tropical algebra:
\begin{equation}
\eqname{T$k$}
\begin{split}
    s \oplus t &= \texttt{largest}(s \cup t, k)\\
    s \odot t &= \texttt{largest}(\{a+b \mid a \in s, b\in t\}, k)\\
    \mymathbb{0} &= -\infty^{\otimes k}\\
    \mymathbb{1} &= -\infty^{\otimes k-1} \otimes 0
\end{split}\label{eq:ext-tropical}
\end{equation}
where $\texttt{largest}(s, k)$ means truncating the set $s$ by only keeping its $k$ largest values.

\begin{theorem}
    Let $G = (V, E)$ be a weighted graph. Its $K$ largest independent set sizes can be computed in time $O(K\log K)\cc$.
\end{theorem}
\begin{proof}
The $K$ largest independent sets can be computed by contracting a tensor network with extended tropical numbers. The vertex tensor and edge tensor are
\begin{equation}
    \left(W^{{\rm T}k}\right)^{(v)} = \left(\begin{matrix}
        \mymathbb{1} \\
        -\infty^{\otimes k-1} \otimes w_v
    \end{matrix}\right),   
    \qquad
        \left(B^{{\rm T}k}\right)^{(u, v)} = \left(\begin{matrix}
        \mymathbb{1}  & \mymathbb{1} \\
        \mymathbb{1} & \mymathbb{0}
    \end{matrix}\right),
\end{equation}
where we have used $x_v^{w_v} = (-\infty^{\otimes k-1} \otimes 1)^{w_v} = -\infty^{\otimes k-1} \otimes w_v$.

The computation of $s \odot t$ is a maximum sum combination problem that can be done in time $O(k\log(k))$ using the algorithm in \Cref{sec:maxsum}. Hence, the overall time complexity of contracting the tensor network is $O(K\log K)\cc$.
\end{proof}

\begin{corollary}
    Let $G = (V, E)$ be a weighted graph. Its $K$ largest independent sets can be computed in time $O(|V|* K\log K)\cc$.
\end{corollary}
To find solutions corresponding to the largest $K$ sizes, one can combine the extended tropical algebra with the bit-string algebra (\Cref{eq:singleconfig}).
Since the $\oplus$ operation of the configuration sampler is not used in the combined algebra, the resulting configurations are deterministic and complete.

\section{Example applications} \label{sec:examples}

\subsection{Number of independent sets and entropy constant for hardcore lattice gases}\label{sec:entropy}
We compute the counting of all independent sets for graphs shown in \Cref{fig:lattices}, where vertices are all placed on square lattices of dimensions $L \times L$.
The types of graphs include: the square lattice graphs (\Cref{fig:lattices}(a)), the square lattice graphs with a filling factor $p=0.8$, which means $\lfloor pL^{2} \rceil$ sites are occupied with vertices  (\Cref{fig:lattices}(b)),
the King's graphs  (\Cref{fig:lattices}(c)), the King's graphs with a filling factor $p = 0.8$  (\Cref{fig:lattices}(d)), which is the ensemble of graphs used in Ref.~\cite{Ebadi2022} to benchmark quantum algorithms on a Rydberg atom array quantum computer. 

The number of independent sets for square lattice graphs of size $L \times L$ form a well-known integer sequence (\href{https://oeis.org/A006506}{OEIS A006506}), which is thought as a two-dimensional generalization of the Fibonacci numbers.
We computed the integer sequence for $L=38$ and $L=39$, which is, to the best of our knowledge, not known before.
In the computation, we used finite-field algebra for contracting integer tensor networks with arbitrarily high precision. 

A theoretically interesting number that can be computed using the number of independent sets is the entropy constant, which can describe the thermodynamic properties of hard-core lattice gases at the high-temperature limit.
For the square lattice graphs, this number is called the \textit{hard square entropy constant} (\href{https://oeis.org/A085850}{OEIS A085850}), which is defined as $\lim_{L\rightarrow \infty} F(L,L)^{1/L^2}$, where $F(L,L)$ is the number of independent sets of a given lattice dimensions $L \times L$.
This quantity arises in statistical mechanics of hard-square lattice gases~\cite{Baxter1980, Pearce1988} and is used to understand phase transitions for these systems. This entropy constant is not known to have an exact representation, but it is accurately known in many digits. Similarly, we can define entropy constants for other lattice gases. In \Cref{fig:hardsquare}, we look at how $F(L,L)^{1/\lfloor pL^2\rceil}$ scales as a function of the grid size $L$ for all types of graphs shown in \Cref{fig:lattices}. Our results match the known results for the non-disordered square lattice and King's graphs. For disordered square lattice and King's graphs with a filling factor $p=0.8$, we randomly sample 1000 graph instances. To our knowledge, the entropy constants for these disordered graphs have not been studied before. They may be used to study phase transitions for disordered lattices, which are typically much harder to understand. Interestingly, the variations due to different random instances are negligible for this quantity. 

\begin{figure}[t] 
    \centering
    \includegraphics[width=\textwidth, trim={0cm 1cm 0cm 1cm}, clip]{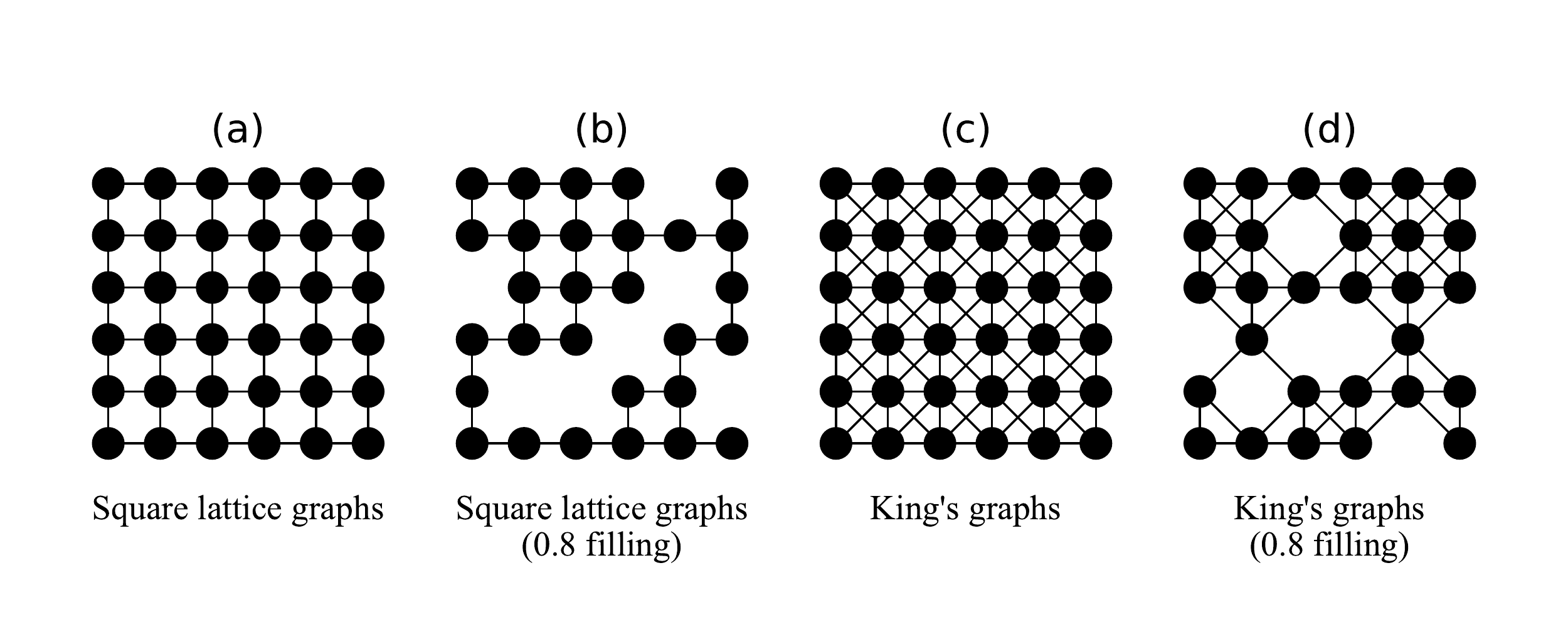}
    \caption{The types of graphs used in the case study in \Cref{sec:entropy}.
    The lattice dimensions are $L\times L$. (a) square lattice graphs. (b) square lattice graphs with a filling factor $p=0.8$.
    (c) King's graphs. (d) King's graphs with a filling factor $p=0.8$.}
    \label{fig:lattices}
\end{figure}

\begin{figure}[t] 
    \centering
    \includegraphics[width=\textwidth, trim={0cm 0cm 0cm 0cm}, clip]{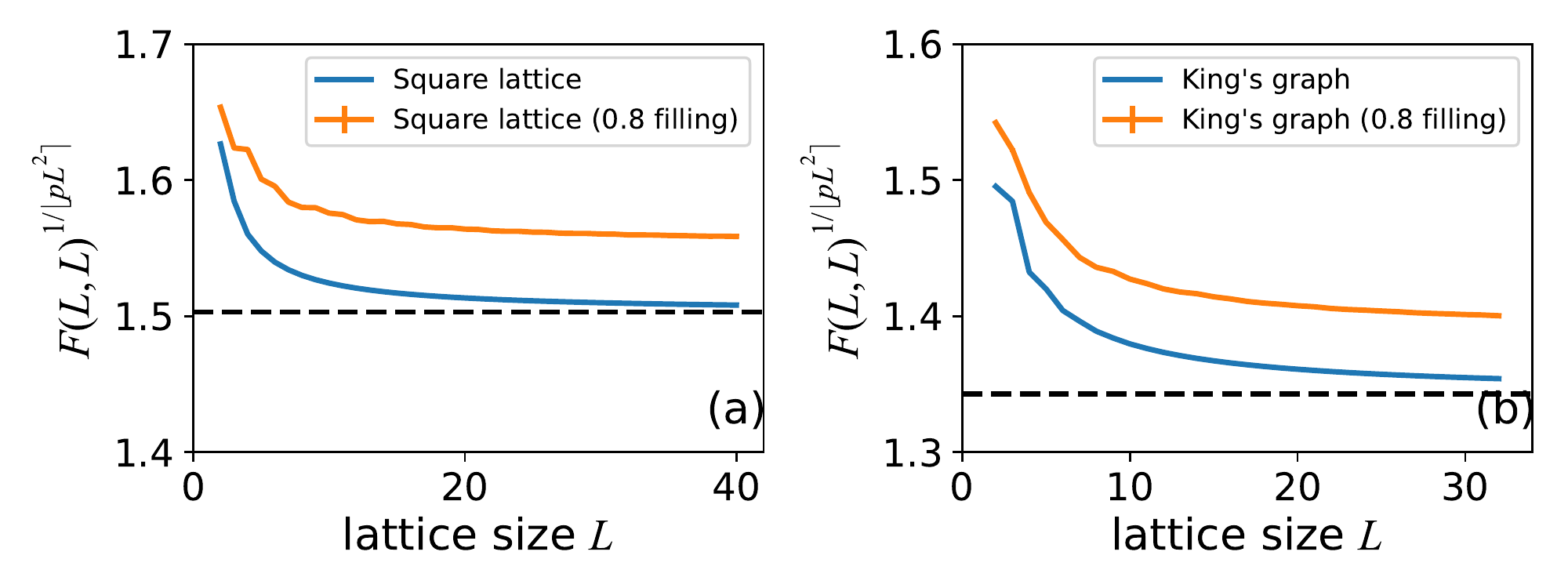}
    \caption{Mean entropy for lattice gases on graphs defined in \Cref{fig:lattices}.
    We sampled $1000$ instances for $p=0.8$ lattices and the error bar is too small to be visible.
    The horizontal black dashed lines are for $\lim_{L\rightarrow \infty} F(L,L)^{1/L^2}$ for the corresponding non-disordered square lattice and King's graphs.
    }
    \label{fig:hardsquare}
\end{figure}

\subsection{The overlap gap property\label{sec:overlap-gap}}
With this tool to enumerate or sample configurations, one can understand the structure of the independent set configuration space, such as the optimization landscape for finding the MISs.
One of the known barriers to finding the MIS is the so-called overlap gap property~\cite{Gamarnik2013, Gamarnik2019}.
If the overlap gap property is present, it means every two large independent sets either have a significant intersection or a very small intersection;
it implies that large independent sets are clustered together.
This clustering property has been used to rigorously prove upper bounds on the performance of local search algorithms~\cite{Gamarnik2013, Gamarnik2019}.
To investigate the overlap gap property, we compute pair-wise Hamming distance distributions of large independent sets as they are good indicators of the presence or absence of overlap gap properties.
We inspect two types of graphs that are particularly interesting, the King's graphs with defects and $3$-regular graphs.
It is known that the MIS problem on a general graph can be mapped to the King's graph with defects~\cite{Garey1977,Ebadi2022}. However, it is not clear whether the MIS problem defined on a randomly generated King's graph with defects can have the overlap gap property.
It is known that finding MISs of a $d$-regular graphs has the overlap gap property~\cite{Rahman2017,Gamarnik2021} when both $d$ and the graph sizes are large, but, it is not known whether, for small $d$, e.g.\ for $3$-regular graphs, this statement remains true.
We randomly generated $9$ instances for each category of King's graph at $0.8$ filling with dimensions $20\times 20$ ($320$ vertices) and $3$-regular graphs with $110$ vertices.
At this problem size, there are too many independent sets to fit into any storage, hence we combine the truncated polynomial and sum-product expression tree to directly sample from the target configuration space.
For each instance $G$, we sample $10^4$ pairs of configurations from the independent sets of sizes $\geq \lceil \gamma \times \alpha(G)\rceil$ and show the pair-wise Hamming distance distribution in \Cref{fig:hamming}.
We observe a clear single peak structure at a fixed distance normalized by the MIS size for the King's graphs, indicating the absence of the overlap gap property in a random King's graph at $0.8$ filling.
Since the MIS problem on an arbitrary graph can be mapped to a King's graph at a certain filling, this result is highly nontrivial. It likely implies 
that the King's graphs with defects mapped from hard MIS instances have a very small measure in the total defected King's graph space. In contrast, 
very different pair-wise Hamming distributions are obtained in \Cref{fig:hamming}(b), where we observed the multiple peak structure when the control parameter $\gamma$ is big enough. It indicates the existence of disconnected clusters in the configuration space of the MIS problem on $3$-regular graphs.
We expect this numerical tool can be used to understand this phenomenon better and to further investigate the graph properties and the geometry of the configuration spaces for a variety of graph instances.
\begin{figure} 
    \includegraphics[width=\textwidth, trim={0.0cm 1cm 0.0cm 0cm}, clip]{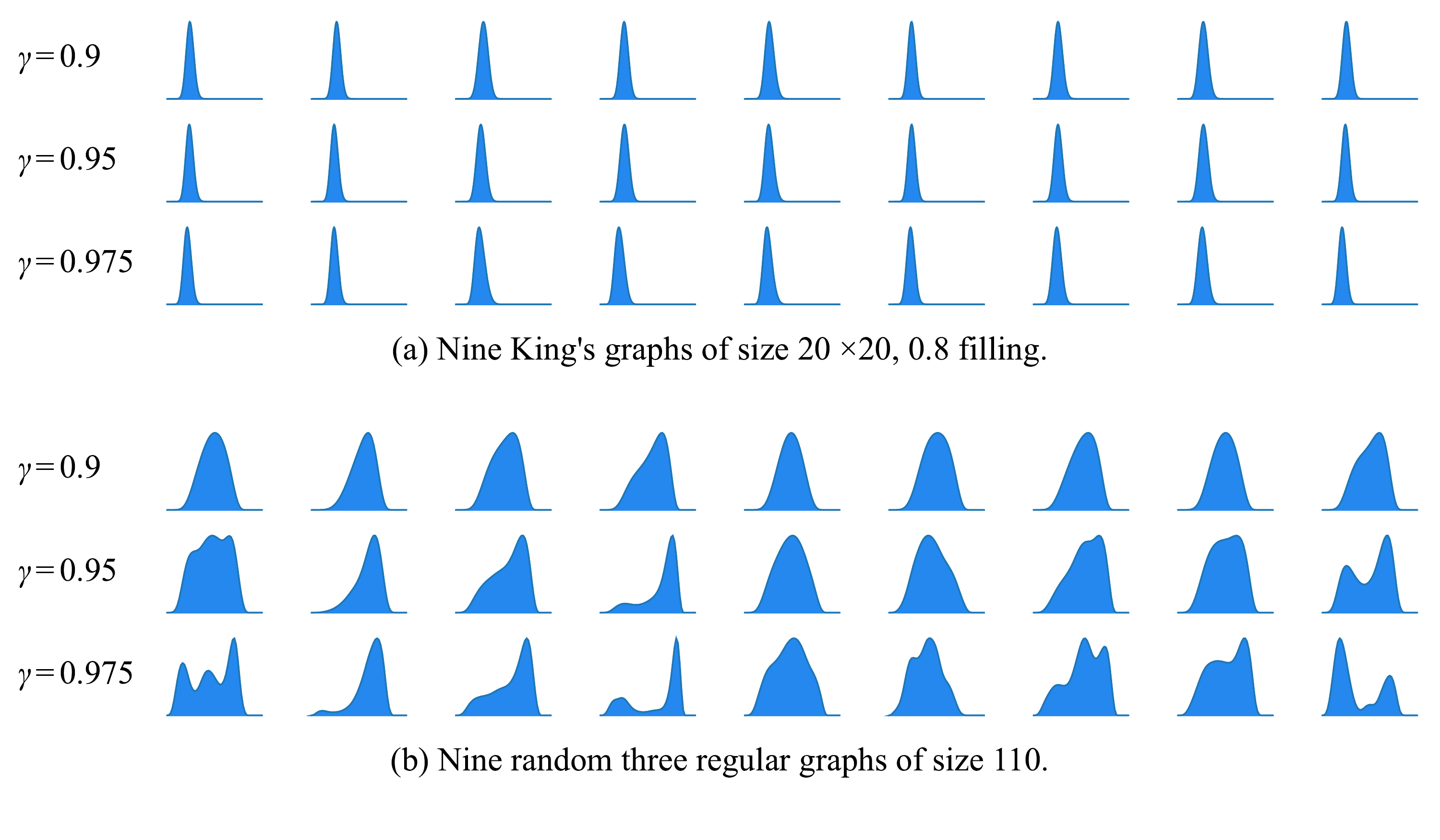}
    \caption{Pairwise Hamming distances distribution for configurations sampled from independent sets with sizes $\geq \lceil\gamma \times \alpha(G)\rceil$.
    In each plot, the $x$-axis is the Hamming distance normalized by the total number of vertices and the $y$-axis is the probability.
    }
    \label{fig:hamming}
\end{figure}

\subsection{Analyzing quantum and classical algorithms for Maximum Independent Set}
\begin{figure} 
    \centering
    \includegraphics[width=.65\textwidth, trim={0cm 0cm 0cm 0cm}, clip]{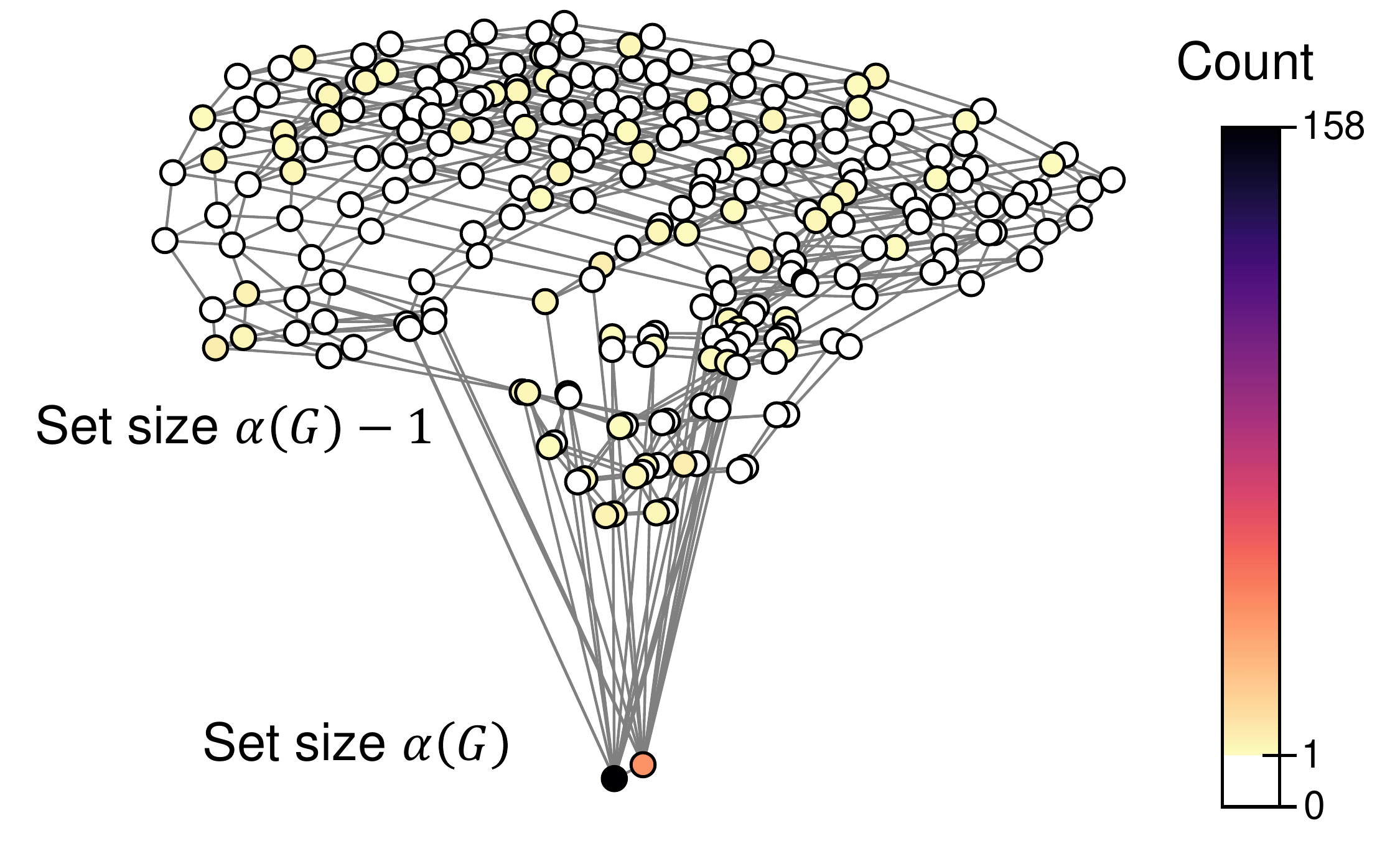}
    \caption{Visualization of experimental outputs of a quantum algorithm for solving the maximum independent set problem~\cite{Ebadi2022}.
    Each vertex represents an independent set, and each edge represents a pair of independent sets that differ by a swap operation or a vertex addition/removal.}
    \label{fig:exp_configuratoins}
\end{figure}
In a recent work, 
the ability to enumerate configurations and compute independence polynomials was critical in understanding the performance of quantum optimization algorithms for the MIS problem on a Rydberg atom quantum computer~\cite{Ebadi2022}. This work focused on 
exploring King's graphs with $0.8$ filling.   The hardest instances for classical simulated annealing could be accurately predicted from the independence polynomial, which gave information about the density of local minima at different independent set sizes.
On the hardest graph instances for simulated annealing studied in the experiment, a high density of local minima were found at independent set sizes of $\alpha(G)-1,$ which the algorithm became trapped in instead of finding the optimal solution of size $\alpha(G)$.
By enumerating the configurations using techniques described in the present work, we found that 
simulated annealing randomly explores the independent sets of size $\alpha(G)-1$ until an optimum solution is found.
Therefore, the large ratio of local to global minima  prevents simulated annealing from efficiently finding an MIS.

Although the performance of the quantum algorithm is more challenging to understand due to the inherent difficulty in studying quantum systems, the present methods allow one 
to gain significant insights by visualizing the experimental outputs of a quantum algorithm over the solution space, as shown in \Cref{fig:exp_configuratoins} for instances with $39$ nodes~\cite{Ebadi2022}.
Here, the structure of the solution space is shown on a graph where each vertex represents a large independent set. Each edge represents a pair of independent sets that differ by a swap operation or a vertex addition to the set, which are the operations naturally present in the effective dynamics at the end of the quantum algorithm. 
The solution space graph is well-connected by local changes to the spin configurations, has a small diameter, and the degree of each node appears to concentrate. 
This visualization makes it clear that 
the quantum algorithm does not appear to return solely local minima with a large Hamming distance from the MISs as suggested for adiabatic algorithms e.g.~by Ref.~\cite{altshuler2010}  which would appear as a long path on the solution space graph from the sampled local minima to the MIS. Instead, the quantum algorithm samples from local minima across the solution space graph with a wide range of Hamming distances from the MISs. In the case where a large superposition state of local minima is created during the coherent evolution, the quantum algorithm achieves a quadratic speedup over simulated annealing~\cite{Ebadi2022}.
Looking forward, we expect these tools can be applied to understanding the performance of quantum and classical algorithms on a wide class of NP-hard combinatorial optimization problems.

\section{Discussion and conclusion}
In this work, we introduced a framework that uses generic tensor networks to compute different solution space properties of a certain class of NP-hard combinatorial optimization problems.
Each solution space property is computed using the same tensor network with different tensor element algebra.
The different data types introduced in the main text to compute these properties are summarized in the diagram in \Cref{fig:venn-diagram}.
The class of problems solvable by a tensor network includes but is not limited to maximum independent sets and a variety of other combinatorial problems such as the matching problem, the k-coloring problem, the max-cut problem, the set packing problem, and the set covering problem, as detailed in \Cref{sec:otherproblems}. 

Looking ahead, it could be possible to generalize the idea of generic programming to other algorithms that have certain algebraic structures such as those using the inclusion-exclusion principle or subset convolution~\cite{Fomin2013} and explore what new properties can be computed.
To this end, dynamic programming~\cite{Courcelle1990, Fomin2013} approaches could be considered.
Dynamic programming is closely related to a tropical tensor network~\cite{Liu2021}; for example, the Viterbi algorithm for finding the most probable configuration in a hidden Markov model can be interpreted as a matrix product state featured with tropical algebra, and the tropical tensor network in the main text is potentially equivalent to dynamic programming in finding an optimum solution.
Since dynamic programming has much broader applications, it would be interesting to extend the ideas from this paper to provide an algebraic interpretation for dynamic programming so that it can be used to compute other solution space properties beyond just finding an optimum solution.
It is also possible to extend this idea to other algebras. For example, generic semiring algebra has been used in computational linguistics to compute interesting quantities of a given grammar and string~\cite{Goodman1999}.

The source code in the Julia language for this paper can be found in the Github repository~\cite{GenericTensorNetworks}. 
There is a short introduction to this repository as well a gist to show how it works in \Cref{sec:technical} as well.
We expect our tool can be used to understand and study many interesting applications of independent sets and beyond.
We also hope the toolkit we built, including tensor network contraction order optimization and efficient tropical matrix multiplication, can be helpful to the development of other scientific software.

\section*{Acknowledgments}
We would like to thank Pan Zhang for sharing his python code for optimizing contraction orders of a tensor network.
We acknowledge Sepehr Ebadi and Leo Zhou for coming up with many interesting questions about independent sets and their questions strongly motivated the development of this project.
We thank Benjamin Schiffer for providing helpful feedback on the writing of this manuscript.
We thank Chris Elord for helping us write the fastest matrix multiplication library for tropical numbers, TropicalGEMM.jl. 
We thank Jacob Miller for helpful discussions. 
We would also like to thank a number of open-source software developers, including Roger Luo, Time Besard, Edward Scheinerman, and Katharine Hyatt
for actively maintaining their packages and resolving related issues voluntarily.
We acknowledge financial support from the DARPA ONISQ program (grant no.\ W911NF2010021), the Center for Ultracold Atoms, the National Science Foundation, the Vannevar Bush Faculty Fellowship, the U.S. Department of Energy (DE-SC0021013 and DOE Quantum Systems Accelerator Center (contract no.\ 7568717), 
the Army Research Office MURI. We acknowledge the computation credits provided by Amazon Web Services for running the benchmarks and case studies. Jinguo Liu acknowledges funding support provided by QuEra Computing Inc.\ through a sponsored research program.

\bibliographystyle{siamplain}
\bibliography{refs}

\appendix

\section{An alternative way to construct the tensor network}\label{sec:energymodel}

Let us characterize the independent set problem on graph $G=(V, E)$ as an energy model with two parts
\begin{equation}\label{eq:eng}
    \mathcal{E}(G, s) = -\sum_{i\in V} w_i s_i + \infty \sum_{(i,j) \in E}s_i s_j
\end{equation}
where $s_i$ is a spin on vertex $i \in V$ and $w_i$ is an onsite energy term associated with it.
The first part corresponds to the negative independent set size and the second part describes the independence constraint, which corresponds to the Rydberg blockade~\cite{Pichler2018, Ebadi2022} in cold atom arrays or the repulsive force in hardcore lattice models~\cite{Dyre2016, Fernandes2007}.
The partition function is defined as
\begin{equation}\label{eq:partition}
    \begin{split}
    Z(G, \beta) = \sum_{s}e^{-\beta \mathcal{E}(G, s)} = \sum_{s\in \mathcal{I}(G)} e^{\beta \sum w_i s_i}\\
         = \sum_{k=0}^{\alpha(G)}a(k) e^{\beta k}  \qquad \quad (k = \sum w_i s_i)
    \end{split}
\end{equation}
where $\mathcal{I}(G)$ is the set of independent sets of graph $G$, $\alpha(G)$ is the absolute value of the minimum energy (maximum independent set size), $a(k)$ is the number of spin configurations with energy $-k$ (independent sets of size $k$).
The partition function can be expressed as a tensor network by placing a vertex tensor on each spin $i \in V$
\begin{equation}
    W^{(i)} = \left(\begin{matrix}
        1 \\
        e^{\beta w_i}
    \end{matrix}\right),
\end{equation}
and an edge tensor on each bond $(u, v) \in E$
\begin{equation}
       B^{(u, v)} = \left(\begin{matrix}
        1  & 1\\
        1 & 0
    \end{matrix}\right),
\end{equation}
where the $0$ in the edge tensor comes from $e^{-\beta\infty}$ in the second term of \Cref{eq:eng}, which is the independence constraint.
By letting $x = e^{\beta}$, we get the tensor network for computing the independence polynomial as described by \Cref{eq:vertextensor} and \Cref{eq:edgetensor}.
If we further let $w_i=1$, the second line of \Cref{eq:partition} is equivalent to the independence polynomial.

\section{Hard problems and tensor networks}\label{sec:otherproblems}
\subsection{Maximal independent sets and maximal cliques}\label{sec:maximal}
In this section, we focus the discussion on the maximal independent sets problem since finding maximal cliques of a graph is equivalent to finding the maximal independent sets of its complement.
Let $G=(V,E)$ be a graph; we denote the neighborhood of a vertex $v\in V$ as $N(v)$.
A maximal independent set $I_m$ is an independent set such that no $v \in V$ satisfies $I_m \cap (\{v\} \cup N[v])  = \emptyset$, i.e.\ an independent set that cannot become a larger one by adding a new vertex.
The corresponding tensor network $\mathcal{N}_{\rm mIS}(G)$ can be specified as
\begin{equation}\label{eq:maxistensornetwork}
\begin{split}
    \Lambda &= \{s_v \mid v \in V\},\\
    \mathcal{T} &= \{T^{(v)}_{s_{N(v,1)} s_{N(v, 2)}\ldots s_{N(v,d(v))}s_v} \mid v\in V\},\\
    \boldsymbol{\sigma}_o &= \varepsilon,\\
\end{split}
\end{equation}
where we defined a tensor for each $v \in V$ and its neighborhood $N(v)$ as
\begin{equation}\label{eq:maximal}
    T^{(v)}_{s_{N(v, 1)}s_{N(v, 2)}\ldots s_{N(v, d(v))}s_v} = \begin{cases}
        s_vx_v^{w_v} & s_{N(v, 1)}=s_{N(v, 2)}=\ldots=s_{N(v, d(v))}=0,\\
        1-s_v& \text{otherwise}.
    \end{cases}
\end{equation}
Here, $N(v, k)$ is the $k$th vertex in $N(v)$ and $d(v) = |N(v)|$ is the degree of vertex $v$.
If $s_v = 1$, then none of its neighbours can be a member of $I_{m}$ by the independence constraint, contributing a factor $x_v^{w_v}$.
If $s_v = 0$, then at least on of its neighbors must be in $I_{m}$ by the maximal constraint, contributing a unit factor.
For a degree 2 vertex $v$, the tensor has the following form
\begin{equation}
    T^{(v)}=\left(\begin{matrix}
    \left(\begin{matrix}
        ~~0 &~1 \\
        ~~1 &~1
    \end{matrix}\right)\\
    \left(\begin{matrix}
        x_v^{w_v} &0 \\
        0 &0
    \end{matrix}\right)
    \end{matrix}\right).
\end{equation}
\begin{theorem}\label{thm:miscomplex}
    The tensor network representation of a maximal independent set problem on a graph $G=(V,E)$ (\Cref{eq:maxistensornetwork}) can be contracted in $cc(\mathcal{N}_{mIS}(G)) = O(|V|)2^{O({\rm tw}(G)\Delta)}$ number of additions and multiplications, where $\Delta$ is the maximum degree of vertices in $G$.
\end{theorem}
\begin{proof}
    The prefactor $|V|$ comes from the number of tensors,
    while the contraction complexity of pair-wise tensor contraction is closely related to the treewidth of the line graph of its hypergraph representation ${\rm tw}(L(\mathcal{N}_{mIS}(G)))$.
    In the following, we will show this quantity is upper bounded by $(\Delta+1)$ times the treewidth of $G$.
    In the line graph $L(\mathcal{N}_{mIS}(G))$, a tensor $T^{(v)}$ corresponds to a clique over $N(v) \cup \{v\}$.
    In the following, we will show given an optimal tree decomposition of $G$, it is always possible to include all cliques into the bags by increasing the bag size by a factor $(\Delta + 1)$.
    Let $v \in V$ be a vertex and $N(v)$ be its neighborhood; we first arbitrarily pick an edge $u \in N(v)$,
    then by the definition of tree decomposition, we can find a bag containing this edge, and lastly we include all $N(v) \cup \{v\}$ into this bag.
    Hence the maximum tensor rank during contraction is upper bounded by ${\rm tw}(G)\Delta $, proving the theorem.
\end{proof}
Let us consider the graph in \Cref{eg:tensorcontraction}. The corresponding tensor network structure for computing the maximal independent polynomial has the following hypergraph representation. \\

    \centerline{\begin{tikzpicture}[
    dot/.style = {circle, fill, minimum size=#1,
                inner sep=0pt, outer sep=0pt},
    dot/.default = 6pt  
                    ]  
        \def\dx{0};
        \def\r{0.4cm}
        \def\G{1.0}
        \foreach \x/\y/\v in {0/0/a, 1/0/b, 2/0/c, 3/0/d, 4/0/e}
            \node[color=black] at (\x*\G+\dx,\y) (\v) {$s_\v$};
        \foreach \x/\v/\t in {0/A/$T_a$, 1/B/$T_b$, 2/C/$T_c$, 3/D/$T_d$, 4/E/$T_e$}
            \node[color=white,fill=black,dot=\r] at (\x*\G+\dx,1.5) (\v) {\scriptsize \t};
        \draw [cyan,thick] (a) -- (A);
        \draw [cyan,thick] (a) -- (B);
        \draw [cyan,thick] (a) -- (C);
        \draw [blue,thick] (b) -- (B);
        \draw [blue,thick] (b) -- (A);
        \draw [blue,thick] (b) -- (C);
        \draw [blue,thick] (b) -- (D);
        \draw [red,thick] (c) -- (C);
        \draw [red,thick] (c) -- (A);
        \draw [red,thick] (c) -- (B);
        \draw [red,thick] (c) -- (D);
        \draw [green,thick] (d) -- (D);
        \draw [green,thick] (d) -- (B);
        \draw [green,thick] (d) -- (C);
        \draw [green,thick] (d) -- (E);
        \draw [orange,thick] (e) -- (E);
        \draw [orange,thick] (e) -- (D);
    \end{tikzpicture}}
 
By contracting this tensor network with generic element types,
we can compute the maximal independent set properties such as the maximal independence polynomial and the enumeration of maximal independent sets.
The maximal independence polynomial is defined as
\begin{equation}
D_{m}(G, x) = \sum_{k=0}^{\alpha(G)} b_k x^k,
\end{equation}
where $b_k$ is the number of maximal independent sets of size $k$.
Comparing with the independence polynomial in \Cref{eq:idpdef}, we have $b_{k} \leq a_{k}$ and $b_{\alpha(G)} = a_{\alpha(G)}$. $D_m(G, 1)$ counts the total number of maximal independent sets~\cite{Gaspers2012, Manne2013};
to our knowledge, the best algorithm has a time complexity $O(1.3642^{|V|})$~\cite{Gaspers2012}.

The benchmark of computing the maximal independent set properties on $3$-regular graphs is shown in \Cref{sec:benchmark}.

\subsection{Matching problem}
A $k$-matching in a graph $G=(V,E)$ is a set of $k$ edges that no two of which have a vertex in common.
We map an edge $(u, v) \in E$ to a degree of freedom $\langle u, v\rangle \in \{0, 1\}$ in a tensor network, where $1$ means an edge is in the set and $0$ otherwise.
The tensor network representation for the matching problem $\mathcal{N}_{\text{match}}$ can be specified as
\begin{equation}\label{eq:matchtensornetwork}
\begin{split}
    \Lambda &= \{\langle u,v\rangle \mid (u, v) \in E\},\\
    \mathcal{T} &= \{W^{(v)}_{\langle v, N(v, 1)\rangle \langle v, N(v,2) \rangle \ldots \langle v, N(v, d(v))\rangle} \mid v\in V\} \cup \{B^{(u, v)}_{\langle u,v\rangle} \mid (u, v) \in E\},\\
    \boldsymbol{\sigma}_o &= \varepsilon,
\end{split}
\end{equation}
where for each $v\in V$, we define a vertex tensor over its neighborhood $N(v)$ as
\begin{equation}
    W^{(v)}_{\langle v, N(v, 1)\rangle \langle v, N(v, 2) \rangle \ldots \langle v, N(v,d(v))\rangle} = \begin{cases}
        1, & \sum_{i=1}^{d(v)} \langle v, N(v, i) \rangle \leq 1,\\
        0, & \text{otherwise},
    \end{cases}
\end{equation}
and for each bond $(u, v) \in E$, we define a rank one tensor as
\begin{equation}
    B^{(u, v)}_{\langle u, v\rangle} = \begin{cases}
    1, & \langle u, v \rangle = 0 \\
    x^{w_{\langle u,v \rangle}}_{\langle u, v\rangle}, & \langle u, v \rangle = 1
\end{cases}.
\end{equation}
Here, $N(v, k)$ is the $k$th vertex in $N(v)$ and $d(v) = |N(v)|$ is the degree of vertex $v$; a label $\langle v, u \rangle$ is equivalent to $\langle u,v\rangle$.
$W$ tensor specifies the constraint that a vertex cannot be shared by two edges in the edge set, and an edge tensor carries the weights.

\begin{theorem}\label{thm:matchcomplex}
    The tensor network representation of a matching problem on graph $G=(V,E)$ (\Cref{eq:matchtensornetwork}) can be contracted in $cc(\mathcal{N}_{\text{match}}(G)) = O(|V|)2^{O({\rm tw}(L(G))}$ number of additions and multiplications, where $L(G)$ is the line graph of $G$.
\end{theorem}
\begin{proof}
    To contract the tensor network, we first absorb edges tensors into the vertex tensors, which does not increase computational complexity.
    After this, the resulting tensor network is isomorphic to $G$.
    Hence the contraction complexity is $O(|V|)2^{O({\rm tw}(L(G))}$.
\end{proof}
Let $x_{\langle u,v\rangle}^{w_{\langle u,v\rangle}}=x$; the tensor network contraction corresponds to the matching polynomial
\begin{equation}
    M(G, x) = \sum\limits_{k=1}^{|V|/2} c_k x^k,
\end{equation}
where $k$ is the size of an edge set, and a coefficient $c_k$ is the number of $k$-matchings.

\subsection{Vertex coloring}
Let $G=(V,E)$ be a graph. A vertex coloring is an assignment of colors to each vertex $v\in V$ such that no edge connects two identically colored vertices. 
In a $k$-coloring problem, the number of colors is limited to less or equal to $k$.
Let us use the 3-coloring problem as an example to show how to reduce it to tensor contractions.
We first map a vertex $v \in V$ to a degree of freedom $c_v\in\{0,1,2\}$.
The tensor network for the vertex coloring problem $\mathcal{N}_{\text{3-color}}(G)$ can be specified as
\begin{equation}\label{eq:colortensornetwork}
\begin{split}
    \Lambda &= \{c_v \mid v \in V\},\\
    \mathcal{T} &= \{W^{(v)}_{c_v} \mid v\in V\} \cup \{B^{(u, v)}_{c_uc_v} \mid (u, v) \in E\},\\
    \boldsymbol{\sigma}_o &= \varepsilon,\\
\end{split}
\end{equation}
where for each vertex $v \in V$, we define a tensor labelled by $c_v$
\begin{equation}
    W^{(v)} = \left(\begin{matrix}
        1_{c_{v} = r}\\
        1_{c_{v} = g}\\
        1_{c_{v} = b}
    \end{matrix}\right),
\end{equation}
and for each edge $(u, v) \in E$, we define a tensor labelled by $(c_u, c_v)$ as
\begin{equation}
    B^{(u, v)} = \left(\begin{matrix}
        0 & x^{w_{uv}} & x^{w_{uv}}\\
        x^{w_{uv}} & 0 & x^{w_{uv}}\\
        x^{w_{uv}} & x^{w_{uv}} & 0
    \end{matrix}\right),
\end{equation}
where subscripts $c_v = r$, $c_v = g$ and $c_v = b$ are for labeling the color configurations. $B$ tensors are for specifying the coloring constraints and $W$ tensors are for labeling the solutions.

\begin{theorem}\label{thm:cutcomplex}
    The tensor network representation of a K-coloring problem on a graph $G=(V,E)$ (\Cref{eq:colortensornetwork}) can be contracted in $cc(\mathcal{N}_{\text{K-color}}(G)) = O(|E|)K^{O({\rm tw}(G))}$ number of additions and multiplications.
\end{theorem}
The proof is similar to that for \Cref{thm:complexreal} except the dimension of each degree of freedom is $K$.
Let $x^{w_{uv}} = x$ and $r_v = g_v = b_v = 1$; we then have a graph polynomial,
in which the $k$-th coefficient is the number of coloring that $k$ bonds satisfy constraints.
If a graph is colorable, the maximum order of this polynomial should be equal to the number of edges in this graph.
Similarly, one can define an edge coloring problem by defining the tensor network on the line graph of $G$.

\subsection{Cutting problem}
In graph theory, a cut is a partition of the vertices of a graph into two disjoint subsets,
which is also known as the spin-glass problem in statistical physics.
Let $G=(V,E)$ be a graph. We associate a weight $w_v$ to each $v\in V$. To reduce the cutting problem on $G$ to the contraction of a tensor network, we first define a Boolean degree of freedom $s_v\in\{0, 1\}$ for each vertex $v\in V$.
The tensor network representation for the cutting problem $\mathcal{N}_{\text{cut}}$ can be specified as
\begin{equation}\label{eq:cuttensornetwork}
\begin{split}
    \mathcal{T} &= \{B^{(u, v)}_{s_us_v} \mid (u, v) \in E\},\\
    \Lambda &= \{s_v \mid v \in V\},\\
    \boldsymbol{\sigma}_o &= \varepsilon,\\
\end{split}
\end{equation}
where for each edge $(u,v)\in E$, we define an edge matrix labelled by $(s_u, s_v)$ as
\begin{equation}
    B^{(u, v)} = \left(\begin{matrix}
        1 & x_{v}^{w_{uv}}\\
        x_{u}^{w_{uv}} & 1
    \end{matrix}\right).
\end{equation}
Here, variables $x_u^{w_{uv}}$ and $x_v^{w_{uv}}$ are for a cut on this edge or a domain wall in a spin glass problem.
\begin{theorem}\label{thm:colorcomplex}
    The tensor network representation of a cuting problem on a graph $G=(V,E)$ (\Cref{eq:cuttensornetwork}) can be contracted in $cc(\mathcal{N}_{\text{cut}}(G)) = O(|E|)2^{O({\rm tw}(G))}$ number of additions and multiplications.
\end{theorem}
The proof is similar to that for \Cref{thm:complexreal}.
Let $x_u^{w_{uv}} = x_v^{w_{uv}} = x$; we have a graph polynomial similar to the previous ones,
in which the $k$th coefficient is two times the number of cut configurations that have size $k$ (i.e.\ cutting $k$ edges).

\subsection{Dominating Set}
In graph theory, a dominating set for a graph $G = (V, E)$ is a subset $D \subseteq V$ such that every vertex not in $D$ is adjacent to at least one member of $D$.
To reduce this problem to the contraction of a tensor network, we first map a vertex $v\in V$ to a Boolean degree of freedom $s_v\in\{0, 1\}$.
The tensor network for the dominating set problem $\mathcal{N}_{\text{dom}}$ can be specified as
\begin{equation}\label{eq:domtensornetwork}
\begin{split}
    \Lambda &= \{s_v \mid v \in V\},\\
    \mathcal{T} &= \{T^{(v)}_{s_{N(v,1)}s_{N(v, 2)} \ldots s_{N(v,d(v))} s_v} \mid v\in V\},\\
    \boldsymbol{\sigma}_o &= \varepsilon,\\
\end{split}
\end{equation}
where for each vertex $v$, we define a tensor on its closed neighborhood $\{v\} \cup N(v)$ as
\begin{equation}
T^{(v)}_{s_{N(v, 1)}s_{N(v, 2)}\ldots s_{N(v, d(v))}s_v} = \begin{cases}
    0 & s_{N(v, 1)}=s_{N(v, 2)}=\ldots=s_{N(v, d(v))}=s_v=0,\\
    1 & s_v=0,\\
    x^{w_v}_v & \text{otherwise}.
\end{cases}
\end{equation}
Here, $w_v$ is the weight associated with the vertex $v$, $N(v, k)$ is the $k$th vertex in $N(v)$ and $d(v) = |N(v)|$ is the degree of vertex $v$.
This tensor implies a configuration having a closed neighborhood of $v$ not in $D$ ($s_{N(v, 1)}=s_{N(v, 2)}=\ldots=s_{N(v, d(v))}=s_v=0$) cannot be a dominating set.
Otherwise, if $v$ is in $D$, this tensor contributes a multiplicative factor $x_v^{w_v}$ to the output. 
\begin{theorem}\label{thm:domcomplex}
    The tensor network representation of a dominating set problem on a graph $G=(V,E)$ (\Cref{eq:domtensornetwork}) can be contracted in $cc(\mathcal{N}_{\text{dom}}(G)) = O(|V|)2^{O({\rm tw}(G)\Delta)}$ number of additions and multiplications.
\end{theorem}
The proof is similar to that for \Cref{thm:miscomplex}.
The graph polynomial for the dominating set problem is known as the domination polynomial~\cite{Alikhani2009}
\begin{equation}
D(G, x) = \sum_{k=0}^{\gamma(G)} d_k x^k,
\end{equation}
where $d_k$ is the number of dominating sets of size $k$.

\subsection{Boolean satisfiability Problem}
The Boolean satisfiability problem is the problem of determining if there exists an assignment that satisfies a given Boolean formula.
One can specify a satisfiable problem in the conjunctive normal form (CNF), i.e.\ a conjunction of clauses (or disjunctions of Boolean literals).
Given the alphabet of Boolean variables $V$ and its negation $\neg V = \{\neg v \mid v \in V\}$, a CNF can be formally defined as
\begin{equation}
    {\rm CNF} = \bigwedge_{k=1}^{M} \bigvee_{i=1}^{|C^{(k)}|} C^{(k)}_i
\end{equation}
where $M$ is the number of clauses, $C^{(k)}$ is the $k$th clause and $C^{(k)}_i \in V \cup \neg V$ is the $i$th literal in it.
The standard tensor network can be used to study the counting version of the satisfiability problem~\cite{Biamonte2015}, while in the following, we will show a generic reduction from the problem of solving a CNF to a tensor network contraction for solving more solution space properties.
We first map each Boolean literal $v\in V$ to a Boolean degree of freedom $s_v \in \{0, 1\}$.
$s_v = 0$ stands for variable $v$ having value \texttt{false} while $s_a=1$ stands for having value \texttt{true}.
The tensor network $\mathcal{N}_{\text{CNF}}$ can be specified as
\begin{equation}\label{eq:sattensornetwork}
\begin{split}
    \Lambda &= \{s_v \mid v \in V\},\\
    \mathcal{T} &= \{W^{(v)}_v \mid v \in V\} \cup \{T^{(k)}_{s_{N(k,1)} s_{N(k, 2)}\ldots s_{N(k,d(k)))}} \mid k=1,\ldots,M\},\\
    \boldsymbol{\sigma}_o &= \varepsilon,\\
\end{split}
\end{equation}
where a tensor defined on literal $v \in V$ is
\begin{equation}
    W^{(v)} = \left(\begin{matrix}
    1\\
    1_v
    \end{matrix}\right)
\end{equation}
and a tensor defined on the clause $C^{(k)}$ is
\begin{equation}
T^{(k)}_{s_{N(k, 1)} s_{N(k, 2)} \ldots s_{N(k, d(k))}} = \begin{cases}
    x^{w_k}, &C^{(k)} \text{~is~ satisfied ~by~ } (s_{N(k, 1)}, s_{N(k, 2)}, \ldots, s_{N(k, d(k))})\\
    1, &{\rm otherwise}
\end{cases},
\end{equation}
where $N(k, i) = \begin{cases}\neg C^{(k)}_i, &C^{(k)}_i \in \neg V\\ C^{(k)}_i, & C^{(k)}_i \in V\end{cases}$ is the $i$th literal in $C^{(k)}$ with its negation sign removed and $d(k)$ is the number of boolean variables in it;
$w_k$ is the weight associated with clause $C^{(k)}$.
\begin{theorem}\label{thm:cnfcomplex}
    The tensor network representation of a CNF (\Cref{eq:sattensornetwork}) can be contracted in $cc(\mathcal{N}_{\text{CNF}}) = O(M)2^{O({\rm tw}(H))}$ number of additions and multiplications, where $M$ is the number of clauses,
    $H$ is a hypergraph constructed by mapping a variable $v\in V$ to a vertex and the $k$th clause $C^{(k)}$ to a hyperedge connecting $\{N(k, i)\mid i=1,\ldots, d(k)\}$.
\end{theorem}
This can be proved by showing the hypergraph $H$ is the line graph of $\mathcal{N}_{\text{CNF}}$.
Let $x^{w_k} = x$ and $1_v = 1$; one can get a polynomial, in which the $k$-th coefficient gives the number of assignments that $k$ clauses are satisfied.

\subsection{Set packing}
Suppose one has a finite set $V$ and a list of subsets of $V$, denoted as $S$. Then, the set packing problem asks if some $k$ subsets in $S$ are pairwise disjoint.
It is the hypergraph generalization of the independent set problem, where a set corresponds to a vertex and an element corresponds to a hyperedge.
The generic tensor network for the set packing problem $\mathcal{N}_{\text{pack}}$ also has a similar form as that for the independent set problem
\begin{equation}\label{eq:packtensornetwork}
\begin{split}
    \Lambda &= \{s_\sigma \mid \sigma \in S\},\\
    \mathcal{T} &= \{B^{(v)}_{s_{N(v, 1)} \ldots s_{N(v, d(v))}} \mid v\in V\} \cup \{W^{(\sigma)}_{s_\sigma} \mid \sigma \in S\},\\
    \boldsymbol{\sigma}_o &= \varepsilon,
\end{split}
\end{equation}
where for each $v \in V$, we have the constraints over sets, $N(v) = \{\sigma\mid \sigma\in S \wedge v \in \sigma\}$, that containing it as
\begin{equation}
    B^{(v)}_{N(v, 1) N(v, 2)\ldots N(v, d(v))} = \begin{cases}
        1, & s_{N(v, 1)}+ s_{N(v, 2)} + \ldots + s_{N(v, d(v))}\leq 1,\\
        0, & \text{otherwise}.
    \end{cases}
\end{equation}
and the vertex tensor for each $\sigma \in S$
\begin{equation}
    W^{(\sigma)} = \left(\begin{matrix}
        1\\
        x^{w_\sigma}_{\sigma}
    \end{matrix}\right),
\end{equation}
where $N(v, k)$ is the $k$th element in $N(v)$ and $d(v) = |N(v)|$ is the number of elements in $N(v)$.

\begin{theorem}\label{thm:packcomplex}
    The tensor network representation of a set packing problem (\Cref{eq:packtensornetwork}) can be contracted in $cc(\mathcal{N}_{\text{pack}}) = O(|S|)2^{O({\rm tw}(H))}$ number of additions and multiplications, where $|S|$ is the number of sets,
    $H$ is a hypergraph constructed by mapping a set $\sigma\in S$ to a vertex and element $v \in V$ to a hyperedge connecting sets, $N(v)$, that containing it.
\end{theorem}
This can be proved by showing the hypergraph $H$ is the line graph of $\mathcal{N}_{\text{pack}}$.

\subsection{Set covering}
Suppose one has a finite set $V$ and a list of subsets of $V$, denoted as $S$.
The set covering problem aims to find the minimum number of sets in $S$ that incorporate (cover) all of elements in $V$.
To get the generic tensor network representation, we first map a set $\sigma \in S$ to a Boolean degree of freedom $s_\sigma\in\{0, 1\}$.
Then the tensor network representation for the set covering problem $\mathcal{N}_{\text{cover}}$ can be specified as
\begin{equation}\label{eq:covertensornetwork}
\begin{split}
    \Lambda &= \{s_\sigma \mid \sigma \in \mathcal{S}\},\\
    \mathcal{T} &= \{W^{(\sigma)}_{s_\sigma} \mid \sigma \in S\} \cup \{B^{(v)}_{s_{N(v, 1)}s_{N(v, 2)}\ldots s_{N(\sigma, d(v))}} \mid v \in V\},\\
    \boldsymbol{\sigma}_o &= \varepsilon,
\end{split}
\end{equation}
where for each $s_\sigma$, we define a parameterized rank-one tensor indexed by it as
\begin{equation}
W^{(\sigma)} = \left(\begin{matrix}
    1 \\
    x_\sigma^{w_\sigma}
    \end{matrix}\right)
\end{equation}
where $x_\sigma$ is a generic typed variable associated with $\sigma$ and $w_\sigma$ is a positive integer as the weight.
For each element $v \in V$, we can define a constraint over all $s \in S$ containing this element, i.e. $N(v) = \{s \mid s \in S \land v\in s\}$, as
\begin{equation}
B^{(v)}_{s_{N(v, 1)} s_{N(v, 2)} \ldots s_{N(v, d(v))}} = \begin{cases}
    0, & s_{N(v, 1)}=s_{N(v, 2)}=\ldots=s_{N(v, d(v))}=0,\\
    1, & \text{otherwise}.
\end{cases}
\end{equation}
where $N(v, k)$ is the $k$th element in $N(v)$ and $d(v) = |N(v)|$ is the number of elements in $N(v)$. 
If a subset of $S$ does not include any sets containing element $v$, then the corresponding entry is zero.

\begin{theorem}\label{thm:covercomplex}
    The tensor network representation of a set covering problem (\Cref{eq:covertensornetwork}) can be contracted in $cc(\mathcal{N}_{\text{cover}}) = O(|S|)2^{O({\rm tw}(H))}$ number of additions and multiplications, where $|S|$ is the number of sets,
    $H$ is a hypergraph constructed by mapping a set $\sigma\in S$ to a vertex and element $v \in V$ to a hyperedge connecting sets $N(v)$, that containing it.
\end{theorem}
This can be proved by showing the hypergraph $H$ is the line graph of $\mathcal{N}_{\text{cover}}$.

\section{Bounding the MIS enumeration space}\label{sec:bounding}
\begin{figure}
    \centering
    \includegraphics[width=0.9\textwidth, trim={0cm 0cm 0cm 0cm}, clip]{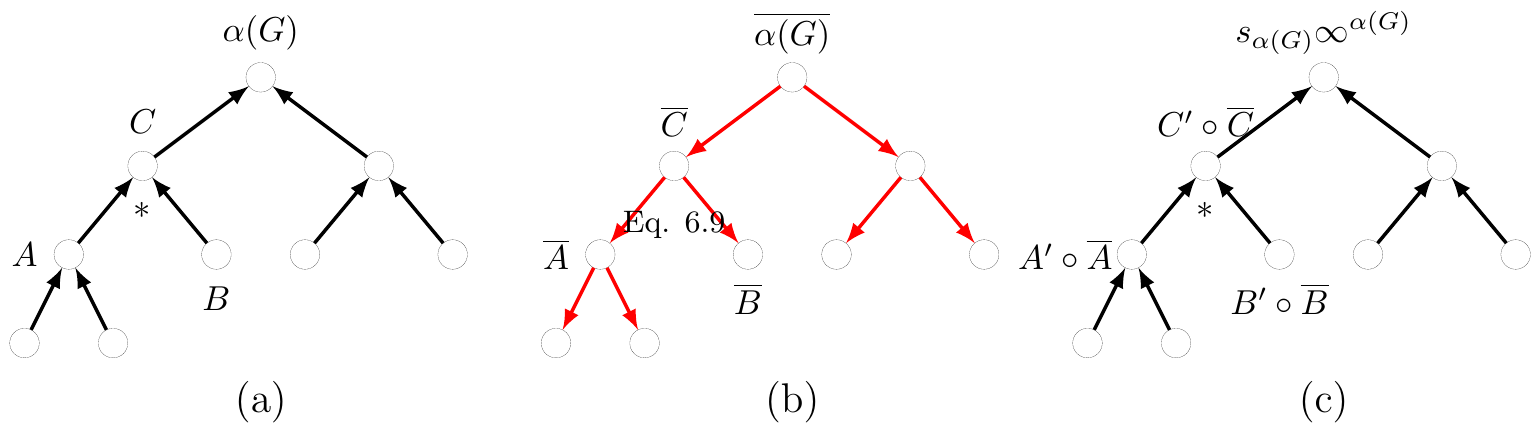}
    \caption{Bounded enumeration of maximum independent sets. Here, a circle is a tensor, an arrow specifies the execution direction of a function, $\overline A$ is the Boolean mask for $A$ and $\circ$ is the Hadamard (element-wise) multiplication. (a) is the forward pass with tropical algebra (\Cref{eq:tropical}) for computing $\alpha(G)$.
     (b) is the backward pass for computing Boolean gradient masks.
     (c) is the masked tensor network contraction with tropical algebra combined with sets (\Cref{eq:countingtropicalset}) for enumerating configurations.}
     \label{fig:bounding}
\end{figure}

When using the algebra in \Cref{eq:countingtropicalset} to enumerate all MISs, the program often stores significantly more intermediate configurations than necessary.
To reduce the space overhead, we will show how to bound the searching space using the MIS size $\alpha(G)$.
The bounded contraction consists of three stages as shown in \Cref{fig:bounding}. (a) We first compute the value of $\alpha(G)$ with tropical algebra and cache all intermediate tensors.
(b) Then, we compute a Boolean mask for each cached tensor, where we use a Boolean \texttt{true} to represent a tensor element having a contribution to the MIS and Boolean \texttt{false} otherwise.
(c) Finally, we perform masked tensor network contraction (i.e.\ discarding elements masked \texttt{false}) using the element type with the algebra in \Cref{eq:countingtropicalset} to obtain all MIS configurations.
The crucial part is computing the masks in step (b). Note that these masks correspond to tensor elements with non-zero gradients to the MIS size; we can compute these masks by back-propagating the gradients.
To derive the back-propagation rule for tropical tensor contraction,
we first reduce the problem to finding the back-propagation rule of a tropical matrix multiplication $C = A B$.
Since $ O_{ik} = \bigoplus_{j} \ A_{ij} \odot B_{jk} = \max_{j} \ A_{ij} \odot B_{jk}$ with tropical algebra, we have the following inequality
\begin{equation}
    A_{ij} \odot B_{jk} \leq C_{ik}.
\end{equation}
Here $\leq$ on tropical numbers are the same as the real-number algebra.
The equality holds for some $j'$, which means $A_{ij'}$ and $B_{j'k}$ have contributions to $C_{ik}$.
Intuitively, one can use this relation to identify elements with nonzero gradients in $A$ and $B$,
but if doing this directly, one loses the advantage of using BLAS libraries~\cite{TropicalGEMM} for high performance.
Since $A_{ij} \odot B_{jk} = A_{ij} + B_{jk}$, one can move $B_{jk}$ to the right hand side of the inequality: 
\begin{equation}
    A_{ij} \leq C_{ik} \odot B_{jk}^{\circ -1}
\end{equation}
where ${}^{\circ -1}$ is the element-wise multiplicative inverse on tropical algebra (which is the additive inverse on real numbers).
The inequality still holds if we take the minimum over $k$: 
\begin{equation}
    A_{ij} \leq \min_{k}(C_{ik} \odot B_{jk}^{\circ -1}) = \left(\max_{k} \left(C_{ik}^{\circ -1} \odot B_{jk} \right) \right)^{\circ -1} = \left(\bigoplus_{k} \left(C_{ik}^{\circ -1} \odot B_{jk} \right) \right)^{\circ -1} = \left( C^{\circ-1} B^{\mathsf{T}} \right)^{\circ -1}_{ij}.
\end{equation}
On the right-hand side, we transform the operation into a tropical matrix multiplication so that we can utilize the fast tropical BLAS routines~\cite{TropicalGEMM}.
Again, the equality holds if and only if the element $A_{ij}$ has a contribution to $C$ (i.e.\ having a non-zero gradient).
Let the gradient mask for $C$ be $\overline C$; the back-propagation rule for gradient masks reads
\begin{equation}\label{eq:adrule}
\overline{A}_{ij} = \delta \left(A_{ij}, \left( \left( C^{\circ-1} \circ \overline C \right) B^{\mathsf{T}} \right)_{ij}^{\circ -1} \right),
\end{equation}
where $\delta$ is the Dirac delta function that returns one if two arguments have the same value and zero otherwise, $\circ$ is the element-wise product, Boolean false is treated as the tropical number $\mymathbb{0}$, and Boolean true is treated as the tropical number $\mymathbb{1}$.
This rule defined on matrix multiplication can be easily generalized to tensor contraction by replacing the matrix multiplication between $C^{\circ-1} \circ \overline C$ and $B^{\mathsf{T}}$ by a tensor contraction.
With the above method, one can significantly reduce the space needed to store the intermediate configurations by setting the tensor elements masked false to zero during contraction.

\section{The fitting approach to computing the independence polynomial}\label{sec:finitefield}
In this section, we propose to find the independence polynomial by fitting $\alpha(G)+1$ random pairs of $x_{i}$ and $y_{i} = I(G,x_{i})$. One can then compute the independence polynomial coefficients $a_{i}$ by solving the linear equation: 
\begin{equation}
\left(\begin{matrix}
1 & x_0 & x_0^2 & \ldots & x_0^{\alpha(G)} \\
1 & x_1 & x_1^2 & \ldots & x_1^{\alpha(G)} \\
\vdots & \vdots & \vdots &\ddots & \vdots \\
1 & x_{\alpha(G)} & x_{\alpha(G)}^2 & \ldots & x_{\alpha(G)}^{\alpha(G)}
\end{matrix}\right)
\left(\begin{matrix}
a_0 \\ a_1 \\ \vdots \\ a_{\alpha(G)}
\end{matrix}\right)
= \left(\begin{matrix}
y_0 \\ y_1 \\ \vdots \\ y_{\alpha(G)}
\end{matrix}\right).\label{eq:lineareq}
\end{equation}
Unlike using the polynomial numbers in \Cref{eq:polynomial},  the fitting approach does not have the linear overhead in space.
However, since the independence polynomial coefficients can have a huge order-of-magnitude range, the round-off errors can be larger than the value itself when using floating-point numbers in computation.
To avoid using the arbitrary precision number that can be very slow and is incompatible with GPU devices, we introduce the following finite-field algebra $\text{GF}(p)$ approach:
\begin{equation}
\eqname{GF$(p)$}
\begin{split}
    x ~\oplus~ y &= x+y\pmod p,\\
    x ~\odot~ y &= xy\pmod p,\\
    \mymathbb{0} &= 0,\\
    \mymathbb{1} &= 1.
\end{split}\label{eq:finitefield}
\end{equation}
Regarding the finite-field algebra, we have the following observations:
\begin{enumerate}
    \item One can use Gaussian elimination~\cite{Golub2013} to solve the linear equation \Cref{eq:lineareq} since it is a generic algorithm that works for any elements with field algebra. The multiplicative inverse of a finite-field algebra can be computed with the extended Euclidean algorithm.
    \item Given the remainders of a larger unknown integer $x$ over a set of co-prime integers $\{p_1, p_2, \ldots, p_n\}$,
    $x \pmod {p_1 \times p_2 \times \ldots \times p_n}$ can be computed using the Chinese remainder theorem. With this, one can infer big integers from small integers.
\end{enumerate}
With these observations, we develop \Cref{alg:finitefield} to compute the independence polynomial exactly without introducing space overheads.
The algorithm iterates over a sequence of large prime numbers until convergence.
In each iteration, we choose a large prime number $p$, and contract the tensor networks to evaluate the polynomial for each variable $\chi = (x_{0}, x_{1}, \ldots, x_{\alpha(G)})$ on ${\rm GF}(p)$ and denote the outputs as $(y_0, y_1, \ldots, y_{\alpha(G)}) \pmod p$.
Then we solve \Cref{eq:lineareq} using Gaussian elimination on ${\rm GF}(p)$ to find the coefficient modulo $p$, $A_p \equiv (a_0, a_1, \ldots, a_{\alpha(G)})\pmod p$.
As the last step of each iteration, we apply the Chinese remainder theorem to update $A \pmod P $ to $ A \pmod {P\times p}$, where $P$ is a product of all prime numbers chosen in previous iterations.
If this number does not change compared with the previous iteration, it indicates the convergence of the result and the program terminates.
All computations are done with integers of fixed-width $W$ except the last step of applying the Chinese remainder theorem, where we use arbitrary precision integers to represent the counting.

\LinesNumberedHidden
\begin{algorithm}[!ht]
    \small
    \SetAlgoNoLine
    Let $P = 1$, $W$ be the integer width, vector $\chi = (0,1,2, \ldots, \alpha(G))$, matrix $X_{ij} = (\chi_i)^j$, where $i,j = 0, 1, \ldots, \alpha(G)$\;

    \While{true}{
        compute the largest prime $p$ that $\gcd(p, P) = 1$ and $p < 2^W$\;

        \For{$i=0\ldots\alpha(G)$}{
            $y_i \pmod p$ = ${\rm contract\_tensor\_network}(\chi_i\pmod p)$ \tcp*[l]{on $\text{GF}(p)$}
        }

        $A_p = (a_0, a_1, \ldots, a_{\alpha(G)}) \pmod p = {\rm gaussian\_elimination}(X, (y_0, y_1, \ldots, y_{\alpha(G)}) \pmod p) $\;

        $A_{P\times p} = {\rm chinese\_remainder}(A_P, A_p)$\;

        \If{$A_P = A_{P \times p}$}{
            \Return $A_P$ \tcp*[l]{converged}
        }
        $P = P \times p$\;
    }\caption{Computing the independence polynomial exactly without integer overflow}\label{alg:finitefield} 
\end{algorithm}

Alternatively, one can use a faster but less accurate Fourier transformation based method to fit this polynomial, which is detailed and benchmarked in \Cref{sec:benchmark}.

\section{Integer sequence formed by the number of independent sets}

We computed the number of independent sets on square lattices and King's graphs with our generic tensor network contraction algorithm on GPUs.
The tensor element type is the finite-field algebra so that we can reach an arbitrary precision.
We also computed the independence polynomial for these lattices up to size $30\times 30$ in our \href{https://github.com/GiggleLiu/NoteOnTropicalMIS/tree/master/data}{GitHub repository}.

\begin{table}[h]
\caption{The number of independent sets for square lattice graphs of size $L\times L$. This forms the integer sequence \href{https://oeis.org/A006506}{OEIS A006506}.
Here we only include two updated entries for $L=38,39$, which, to our knowledge, has not been computed before~\cite{Butera2014}.
}
\begin{center}
\scalebox{0.9}{
\begin{tabular}{|c| >{\centering\arraybackslash} p{0.95\linewidth}|}
 \hline $L$  & square lattice graphs \\
 \hline $38$ & 616 412 251 028 728 207 385 738 562 656 236 093 713 609 747 387 533 907 560 081 990 229 746 115 948 572 583 817 557 035 128 726 922 565 913 748 716 778 414 190 432 479 964 245 067 083 441 583 742 870 993 696 157 129 887 194 203 643 048 435 362 875 885 498 554 979 326 352 127 528 330 481 118 313 702 375 541 902 300 956 879 563 063 343 972 979\\
 \hline $39$ &  29 855 612 447 544 274 159 031 389 813 027 239 335 497 014 990 491 494 036 487 199 167 155 042 005 286 230 480 609 472 592 158 583 920 411 213 748 368 073 011 775 053 878 033 685 239 323 444 700 725 664 632 236 525 923 258 394 737 964 155 747 730 125 966 370 906 864 022 395 459 136 352 378 231 301 643 917 282 836 792 261 715 266 731 741 625 623 207 330 411 607\\
  \hline
\end{tabular}
}
\end{center}
\label{tbl:squaregrid}
\end{table}

\section{Performance benchmarks}\label{sec:benchmark}
We run a single thread benchmark on central processing units (CPU) Intel(R) Xeon(R) CPU E5-2686 v4 @ 2.30GHz, and its CUDA version on a GPU Tesla V100-SXM2 16G.
The results are summarized in \Cref{fig:benchmark}.
The graphs in all benchmarks are random three-regular graphs, which have treewidth that is asymptotically smaller than $|V|/6$~\cite{Fomin2006}.
In this benchmark, we do not include traditional algorithms for finding the MIS sizes such as branching~\cite{Tarjan1977, Robson1986} or dynamic programming~\cite{Courcelle1990, Fomin2013}.
To the best of our knowledge, these algorithms are not suitable for computing most of the solution space properties mentioned in this paper.
The main goal of this section is to show the relative computation time for calculating different solution space properties.

\begin{figure} 
    \centering
    \includegraphics[width=\textwidth, trim={0cm 0cm 0cm 0cm}, clip]{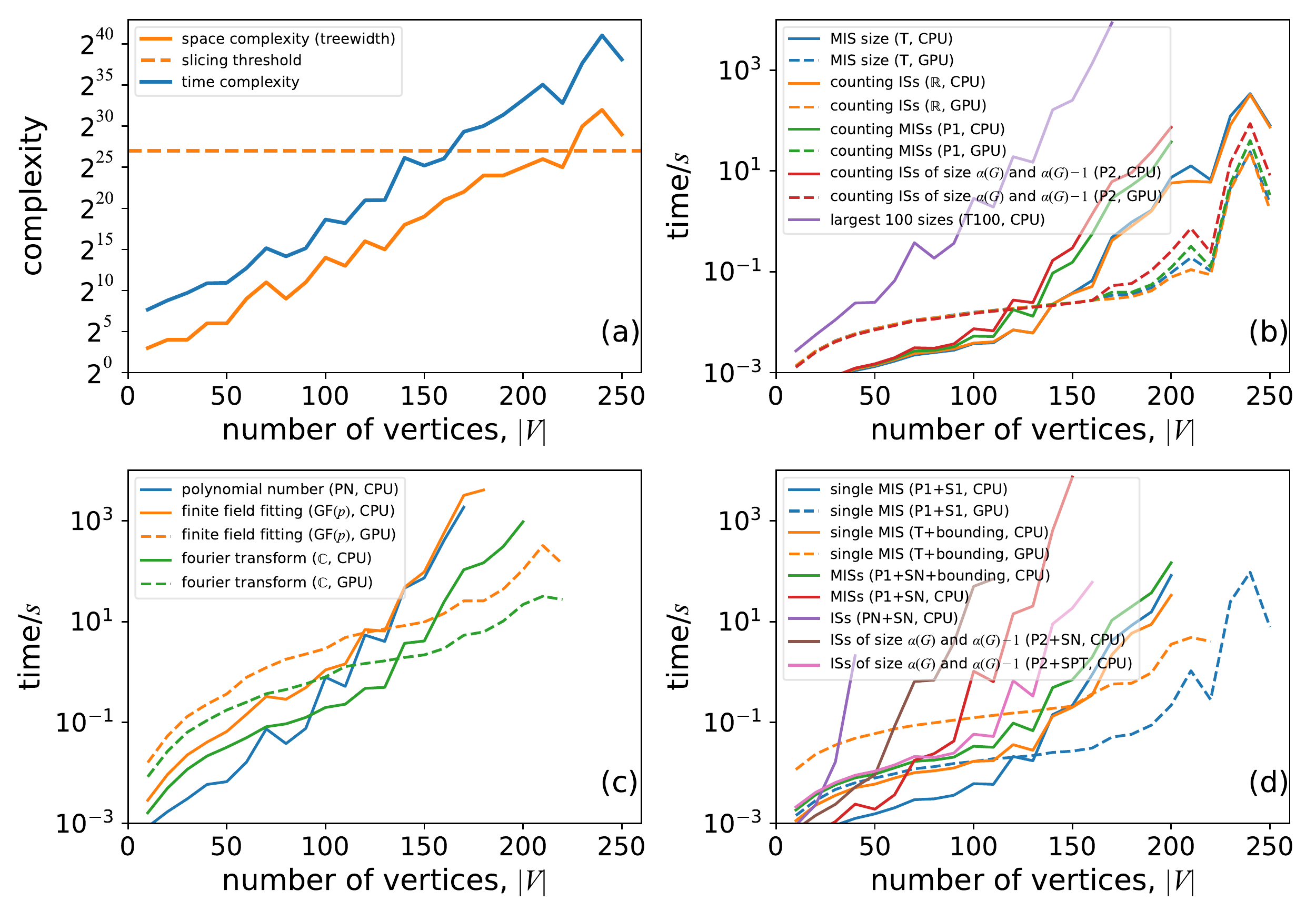}
    \caption{Benchmark results for computing different solution space properties of independent sets of random three-regular graphs with different tensor element types.
    The time in these plots only includes tensor network contraction, without taking into account the contraction order finding and just-in-time compilation time.
    Legends are properties, algebra, and devices that we used in the computation; one can find the corresponding computed solution space property in Table 1 in the main text.
    (a) time and space complexity versus the number of vertices for the benchmarked graphs.
    (b) The computation time for calculating the MIS size and for counting the number of all independent sets (ISs), the number of MISs,
        the number of independent sets having size $\alpha(G)$ and $\alpha(G)-1$, and finding $100$ largest set sizes.
    (c) The computation time for calculating the independence polynomials with different approaches.
    (d) The computation time for configuration enumeration, including single MIS configuration, the enumeration of all independent set configurations, all MIS configurations, all independent sets,
    and all independent set configurations having size $\alpha(G)$ and $\alpha(G)-1$.
    }
    \label{fig:benchmark}
\end{figure}

\Cref{fig:benchmark}(a) shows the time and space complexity of tensor network contraction for different graph sizes.
The contraction order is obtained using the local search algorithm in Ref.~\cite{Kalachev2021}.
If we assume our contraction-order finding program has found the optimal treewidth, which is very likely to be true, the space complexity is the same as the treewidth of the problem graph.
Slicing technique~\cite{Kalachev2021} has been used for graphs with space complexity greater than $2^{27}$ (above the yellow dashed line) to fit the computation into a 16GB memory.
One can see that all the computation times in \Cref{fig:benchmark} (b), (c), and (d) have a strong correlation with the predicted time and space complexity.
While in panel (d), the computation time of configuration enumeration and sum-product expression tree generation also strongly correlates with other factors such as the configuration space size.
Among these benchmarks, computational tasks with data types real numbers, complex numbers, or tropical numbers (CPU only) can utilize fast basic linear algebra subprograms (BLAS) functions.
These tasks usually compute much faster than ones with other element types in the same category.
Immutable data types with no reference to other values can be compiled to GPU devices that run much faster than CPUs in all cases when the problem scale is big enough.
These data types do not include those defined in \Cref{eq:polynomial}, \Cref{eq:set}, \Cref{eq:ext-tropical} and \Cref{eq:expr} or a data type containing them as a part.
In \Cref{fig:benchmark}(c), one can see the Fourier transformation-based method is the fastest in computing the independence polynomial,
but it may suffer from round-off errors (\Cref{sec:fft}). The finite field (GF$(p)$) approach is the only method that does not have round-off errors and can be run on a GPU.
In \Cref{fig:benchmark}(d), one can see the technique to bound the enumeration space in \Cref{sec:bounding} improves the performance for more than one order of magnitude in enumerating the MISs.
The bounding technique can also reduce the memory usage significantly, without which the largest computable graph size is only $\sim150$ on a device with 32GB main memory.

We show the benchmark of computing the maximal independent set properties on $3$-regular graphs in \Cref{fig:benchmark-maximal},
including a comparison to the Bron-Kerbosch algorithm from Julia package \href{https://github.com/JuliaGraphs/Graphs.jl}{Graphs}~\cite{Graphs}.
\Cref{fig:benchmark-maximal}(a) shows the space and time complexities of tensor contraction, which are typically larger than those for the independent set problem.
In \Cref{fig:benchmark-maximal}(b), one can see counting maximal independent sets are much more efficient than enumerating them, while our generic tensor network approach runs slightly faster than the Bron-Kerbosch approach in enumerating all maximal independent sets.

\begin{figure} 
    \centering
    \includegraphics[width=\textwidth, trim={0cm 0cm 0cm 0cm}, clip]{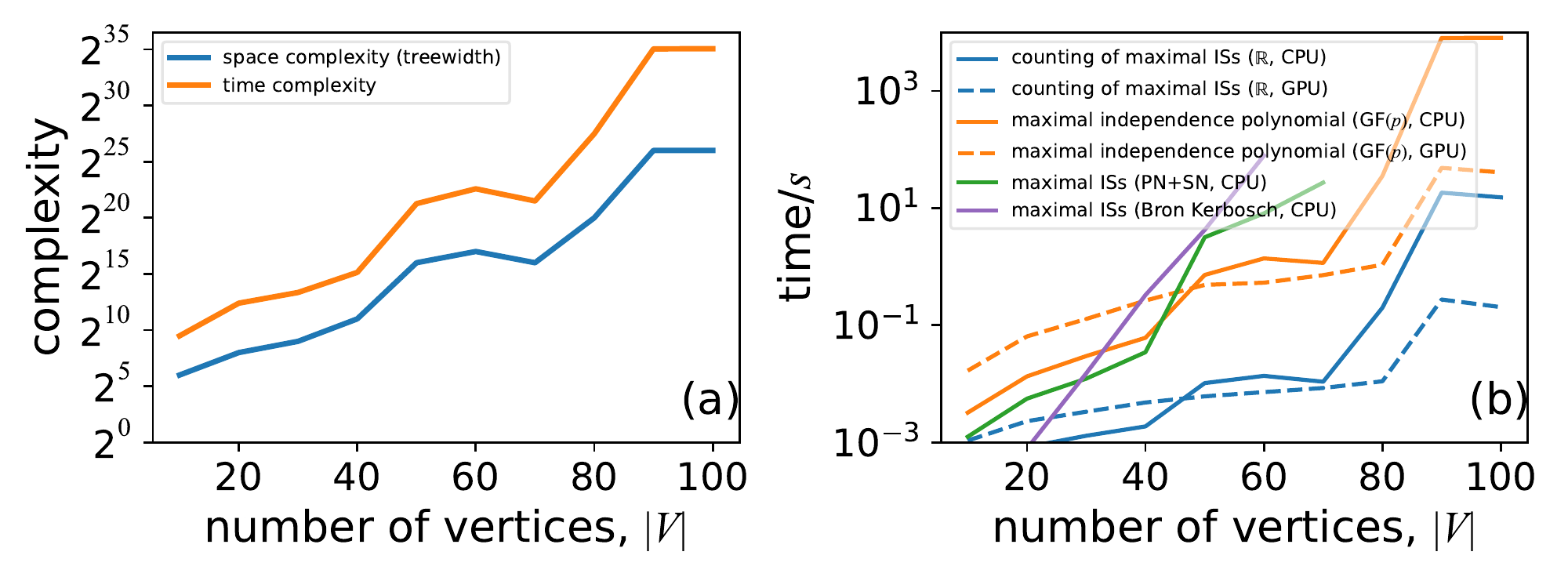}
    \caption{Benchmarks of computing different solution space properties of the maximal independent sets (ISs) problem on random three regular graphs at different sizes.
    (a) time and space complexity of tensor network contraction. 
    (b) The wall clock time for counting and enumeration of maximal ISs. 
    }
    \label{fig:benchmark-maximal}
\end{figure}

\section{An example of increased contraction complexity for the standard tensor network notation}\label{sec:tensorbad}

In the standard Einstein's notation for tensor networks in physics, each index appears precisely twice: either both are in input tensors (which will be summed over) or one is in an input tensor and another in the output tensor.
Hence a tensor network can be represented as an open simple graph, where an input tensor is mapped to a vertex, a label shared by two input tensors is mapped to an edge and a label that appears in the output tensor is mapped to an open edge.
A standard tensor network notation is equivalent to the generalized tensor network in representation power.
A generalized tensor network can be converted to a standard one by adding a $\delta$ tensors at each hyperedge, where a $\delta$ tensor of rank $d$ is defined as
\begin{equation}
    \delta_{i_1, i_2,\ldots,i_d} = \begin{cases}
        1, & i_1=i_2=\ldots =i_d,\\
        0, & \text{otherwise}.
    \end{cases}
\end{equation}

In the following example, we will show this conversion might increase the contraction complexity of a tensor network.
Let us consider the following King's graph. \\

\centerline{\begin{tikzpicture}
    \def\r{0.2}
    \foreach \x in {1,...,5}
        \foreach \y in {1,...,5}
            \filldraw[fill=black] (\x,\y) circle [radius=\r];
    \foreach \x in {1,...,5}
        \foreach \y in {1,...,4}{
            \draw [black,thick] (\x,\y) -- (\x,\y+1);
            \draw [black,thick] (\y,\x) -- (\y+1,\x);
        }
    \foreach \x in {1,...,4}
        \foreach \y in {1,...,4}{
            \draw [black,thick] (\x,\y) -- (\x+1,\y+1);
            \draw [black,thick] (\y+1,\x) -- (\y,\x+1);
        }
\end{tikzpicture}}

The generalized tensor network for solving the MIS problem on this graph has the following hypergraph representation, where we use different colors to distinguish different hyperedges.

\centerline{\begin{tikzpicture}
    \def\r{0.08}
    \def\a{0.07}
    \def\L{0.6}
    \def\l{0.1}
    \def\sql{0.24}
    \pgfmathsetseed{2}
    \foreach[evaluate={\cr=0.1+0.5*Mod(\x,2)}] \x in {1,...,5}
        \foreach[evaluate={\cg=0.1+0.3*Mod(\y,2); \cy=0.5-0.5*Mod(\y,2)}] \y in {1,...,5}{
            \edef\R{\pdfuniformdeviate 255}
            \edef\G{\pdfuniformdeviate 255}
            \edef\B{\pdfuniformdeviate 255}
            \xdefinecolor{MyColor}{RGB}{\R,\G,\B}
            \ifnum \x < 5
                \draw [thick, MyColor, opacity=1.0, line cap=round] (\x,\y) -- (\x+0.5,\y);
                \ifnum \y < 5
                \draw [thick, MyColor, opacity=1.0, line cap=round] (\x,\y) -- (\x+0.3,\y+0.3);
                \fi
                \ifnum \y > 1
                \draw [thick, MyColor, opacity=1.0, line cap=round] (\x,\y) -- (\x+0.3,\y-0.3);
                \fi
            \fi
            \ifnum \x > 1
                \draw [thick, MyColor, opacity=1.0, line cap=round] (\x,\y) -- (\x-0.5,\y);
                \ifnum \y < 5
                \draw [thick, MyColor, opacity=1.0, line cap=round] (\x,\y) -- (\x-0.7,\y+0.7);
                \fi
                \ifnum \y > 1
                \draw [thick, MyColor, opacity=1.0, line cap=round] (\x,\y) -- (\x-0.7,\y-0.7);
                \fi
            \fi
            \ifnum \y < 5
                \draw [thick, MyColor, opacity=1.0, line cap=round] (\x,\y) -- (\x,\y+0.5);
            \fi
            \ifnum \y > 1
                \draw [thick, MyColor, opacity=1.0, line cap=round] (\x,\y) -- (\x,\y-0.5);
            \fi
        }
    \foreach \x in {1,...,5}
        \foreach \y in {1,...,5}{
        }
    \foreach \x in {1,...,5}
        \foreach \y in {1,...,5}{
            \ifnum \y < 5
                \filldraw[fill=black] (\x,\y+0.5) circle [radius=\r];
                \filldraw[fill=black] (\y+0.5,\x) circle [radius=\r];
            \fi
        }
    \foreach \x in {1,...,4}
        \foreach \y in {1,...,4}{
            \filldraw[fill=black] (\x+0.3,\y+0.3) circle [radius=\r];
            \filldraw[fill=black] (\y+0.3,\x+0.7) circle [radius=\r];
        }
    \tikzset{decoration={snake,amplitude=.4mm,segment length=2mm,
                    post length=0mm,pre length=0mm}}
    \draw [decorate] (2.7, 0.5) -- (2.7, 5.5);
\end{tikzpicture}}
Vertex tensors are not shown here because they can be absorbed into an edge tensor and hence do not change the contraction complexity.
If we contract this tensor network in the column-wise order, the maximum intermediate tensor has rank $\sim L$, which can be seen by counting the number of colors at the cut.

By adding $\delta$ tensors to hyperedges, we have the standard tensor network represented as the following simple graph.

\centerline{\begin{tikzpicture}
    \def\r{0.08}
    \def\a{0.1}
    \foreach \x in {1,...,5}
        \foreach \y in {1,...,5}
            \filldraw[fill=black] (\x,\y) circle [radius=\r];
    \foreach \x in {1,...,5}
        \foreach \y in {1,...,4}{
            \filldraw[fill=black] (\x,\y+0.5) circle [radius=\r];
            \filldraw[fill=black] (\y+0.5,\x) circle [radius=\r];
            \draw [black,thick] (\x,\y) -- (\x,\y+1);
            \draw [black,thick] (\y,\x) -- (\y+1,\x);
        }
    \foreach \x in {1,...,4}
        \foreach \y in {1,...,4}{
            \filldraw[fill=black] (\x+0.3,\y+0.3) circle [radius=\r];
            \filldraw[fill=black] (\y+0.3,\x+0.7) circle [radius=\r];
            \draw [black,thick] (\x,\y) -- (\x+1,\y+1);
            \draw [black,thick] (\y+1,\x) -- (\y,\x+1);
        }
    \tikzset{decoration={snake,amplitude=.4mm,segment length=2mm,
                    post length=0mm,pre length=0mm}}
    \draw [decorate] (2.7, 0.5) -- (2.7, 5.5);
\end{tikzpicture}}

In this diagram, the additional $\delta$ tensors can have ranks up to $8$.
If we still contract this tensor network in a column-wise order, the maximum intermediate tensor has rank $\sim3L$, i.e. the space complexity is $\approx 2^{3L}$, which has a larger complexity than using the generalized tensor network notation.

\section{The discrete Fourier transform approach to computing the independence polynomial}\label{sec:fft}

In \Cref{sec:finitefield} in the main text, we show that the independence polynomial can be obtained by solving the linear equation \Cref{eq:lineareq} using the finite field algebra.
One drawback of using finite field algebra is that its matrix multiplication is less computationally efficient compared with floating-point matrix multiplication.
Here, we show an alternative method with standard number types but with controllable round-off errors.
Instead of choosing $x_{i}$ as random numbers, we can choose them such that they form a geometric sequence in the complex domain $x_j = r\omega^j$, where $r \in \mathbb{R}$ and $\omega = e^{-2\pi i/( \alpha(G)+1)}$. The linear equation thus becomes
\begin{equation}
\left(\begin{matrix}
1 & r & r^2 & \ldots & r^{\alpha(G)} \\
1 & r\omega & r^2\omega^2 & \ldots & r^{\alpha(G)} \omega^{\alpha(G)} \\
\vdots & \vdots & \vdots &\ddots & \vdots \\
1 & r\omega^{\alpha(G)} & r^2\omega^{2{\alpha(G)}} & \ldots & r^{\alpha(G)}\omega^{{\alpha(G)}^2}
\end{matrix}\right)
\left(\begin{matrix}
a_0 \\ a_1 \\ \vdots \\ a_{\alpha(G)}
\end{matrix}\right)
= \left(\begin{matrix}
y_0 \\ y_1 \\ \vdots \\ y_{\alpha(G)}
\end{matrix}\right).
\end{equation}

Let us rearrange the coefficients $r^j$ to $a_j$, the matrix on the left side becomes the discrete Fourier transform matrix. Thus, we can obtain the coefficients by inverse Fourier transform $\vec a_r = {\rm FFT^{-1}}(\omega) \cdot \vec y$, where $(\vec a_r)_j = a_j r ^j$.
By choosing different $r$, one can obtain better precision for small $j$ by choosing $r<1$ or large $j$ by choosing $r>1$.

\section{Computing maximum sum combination}\label{sec:maxsum}
Given two sets $A$ and $B$ of the same size $n$.
It is known that the maximum $n$ sum combination of $A$ and $B$ can be computed in time $O(n\log(n))$.
The standard approach to solve the sum combination problem requires storing the variables in a heap --- a highly dynamic binary tree structure that can be much slower to manipulate than arrays.
In the following, we show an algorithm with roughly the same complexity but does not need a heap.
This algorithm first sorts both $A$ and $B$ and then uses the bisection to find the $n$-th largest value in the sum combination.
The key point is we can count the number of entries greater than a specific value
in the sum combination of $A$ and $B$ in linear time.
As long as the data range is not exponentially large, the bisection can be done in $O(\log(n))$ steps, giving the time complexity $O(n\log(n))$.
We summarize the algorithm as in \Cref{alg:sumcombination}.

\LinesNumberedHidden
\begin{algorithm}[!ht]
    \SetKwProg{Fn}{function}{}{end}
    \small
    \SetAlgoNoLine
    Let $A$ and $B$ be two sets of size $n$\;

    \tcp{sort $A$ and $B$ in ascending order}
    $A \gets {\rm sort}(A)$\;

    $B \gets {\rm sort}(B)$\;

    \tcp{use bisection to find the $n$-th largest value in sum combination}

    ${\rm high} \gets A_n+B_n$\;

    ${\rm low} \gets A_1+B_n$\;

    \While{true}{
        ${\rm mid} \gets ({\rm high} + {\rm low}) / 2$\;

        $c \gets \texttt{count\_geq}(n, A, B, {\rm mid})$\;

        \uIf{$c > n$}{
            ${\rm low} \gets {\rm mid}$\;
        }\uElseIf{$c = n$}{
            \Return \texttt{collect\_geq}$(n, A, B, {\rm mid})$\;
        }\Else{
            ${\rm high} \gets {\rm mid}$\;
        }
    }

    \Fn{{\rm \texttt{count\_geq}}($n$, $A$, $B$, $v$)}{
        $k \gets 1$\; \tcp*[l]{number of entries in $A$ s.t. $a+b\geq v$}

        $a \gets A_{n}$\; \tcp*[l]{the smallest entry in $A$ s.t. $a+b\geq v$}

        $c \gets 0$\; \tcp*[l]{the counting of sum combinations s.t. $a+b\geq v$}

        \For{$q$ = $n, n-1 \ldots 1$}{
            $b \gets B_{n-q+1}$\;

            \While{$k < n$ \textbf{and} $a+b \geq v$}{
                $k \gets k+1$\;

                $a \gets A_{n-k+1}$\;
            }

            \uIf{$a+b \geq v$}{
                $c \gets c + k$\;
            }\Else{
                $c \gets c + k-1$\;
            }
        }
        \Return c\;
    }







    \caption{Fast sum combination without using heap}\label{alg:sumcombination} 
\end{algorithm}

In this algorithm, function \texttt{collect\_geq} is similar the \texttt{count\_geq} except the counting is replace by collecting the items to a set.
Inside the function \texttt{count\_geq}, variable $k$ monotoneously increase while $q$ monotoneously decrease in each iteration and
the total number of iterations is upper bounded by $2n$.
Here for simplicity, we do not handle the special element $-\infty$ in $A$ and $B$ and the potential degeneracy in the sums.
It is nevertheless important to handle them properly in a practical implementation.

\section{Technical guides}\label{sec:technical}

This appendix covers some technical guides for efficiency, including an introduction to an open-source package \href{https://github.com/QuEraComputing/GenericTensorNetworks.jl}{GenericTensorNetworks}~\cite{GenericTensorNetworks} implementing the algorithms in this paper and the gist about how this package is implemented.
One can install \texttt{GenericTensorNetworks} in a Julia REPL, by first typing \colorbox{lightgray}{\texttt{]}} to enter the \texttt{pkg>} mode and then typing
\begin{lstlisting}
pkg> add GenericTensorNetworks
\end{lstlisting}
followed by an \colorbox{lightgray}{\texttt{<ENTER>}} key.
To use it for solving solution space properties, just go back to the normal mode (type \colorbox{lightgray}{\texttt{<BACKSPACE>}}) and type

\begin{lstlisting}
julia> using GenericTensorNetworks, Graphs

julia> # using CUDA

julia> solve(
           IndependentSet(
               Graphs.random_regular_graph(20, 3);
               optimizer = TreeSA(),
               weights = NoWeight(),
               openvertices = ()
           ),
           GraphPolynomial();
           usecuda=false
       )
0-dimensional Array{Polynomial{BigInt, :x}, 0}:
Polynomial(1 + 20*x + 160*x^2 + 659*x^3 + 1500*x^4 + 1883*x^5 + 1223*x^6 + 347*x^7 + 25*x^8)
\end{lstlisting}

Here the main function \texttt{solve} takes three inputs: the problem instance of type \texttt{IndependentSet}, the property instance of type \texttt{GraphPolynomial} and an optional key word argument \texttt{usecuda} to decide to use GPU or not.
If one wants to use GPU to accelerate the computation, ``\texttt{using CUDA}'' must be uncommented.
The problem instance takes four arguments to initialize: the only positional argument is the graph instance one wants to solve, the keyword argument \texttt{optimizer} is for specifying the tensor network optimization algorithm, the keyword argument \texttt{weights} is for specifying the weights of vertices as either a vector or \texttt{NoWeight()}, and the keyword argument \texttt{openvertices} is for specifying the degrees of freedom not summed over.
Here, we use the \texttt{TreeSA} method as the tensor network optimizer, and leave \texttt{weights} and \texttt{openvertices} as  default values. The \texttt{TreeSA} algorithm, which was invented in Ref.~\cite{Kalachev2021}, performs the best in most of our applications.
The first execution of this function will be a bit slow due to Julia's just-in-time compilation.
After that, the subsequent runs will be faster.
The following diagram lists possible combinations of input arguments, where functions in the \texttt{Graph} are mainly defined in the package \texttt{Graphs}, and the rest can be found in \texttt{GenericTensorNetworks}.

\centerline{\includegraphics[width=0.8\textwidth, trim={0cm 1cm 0cm 1cm}, clip]{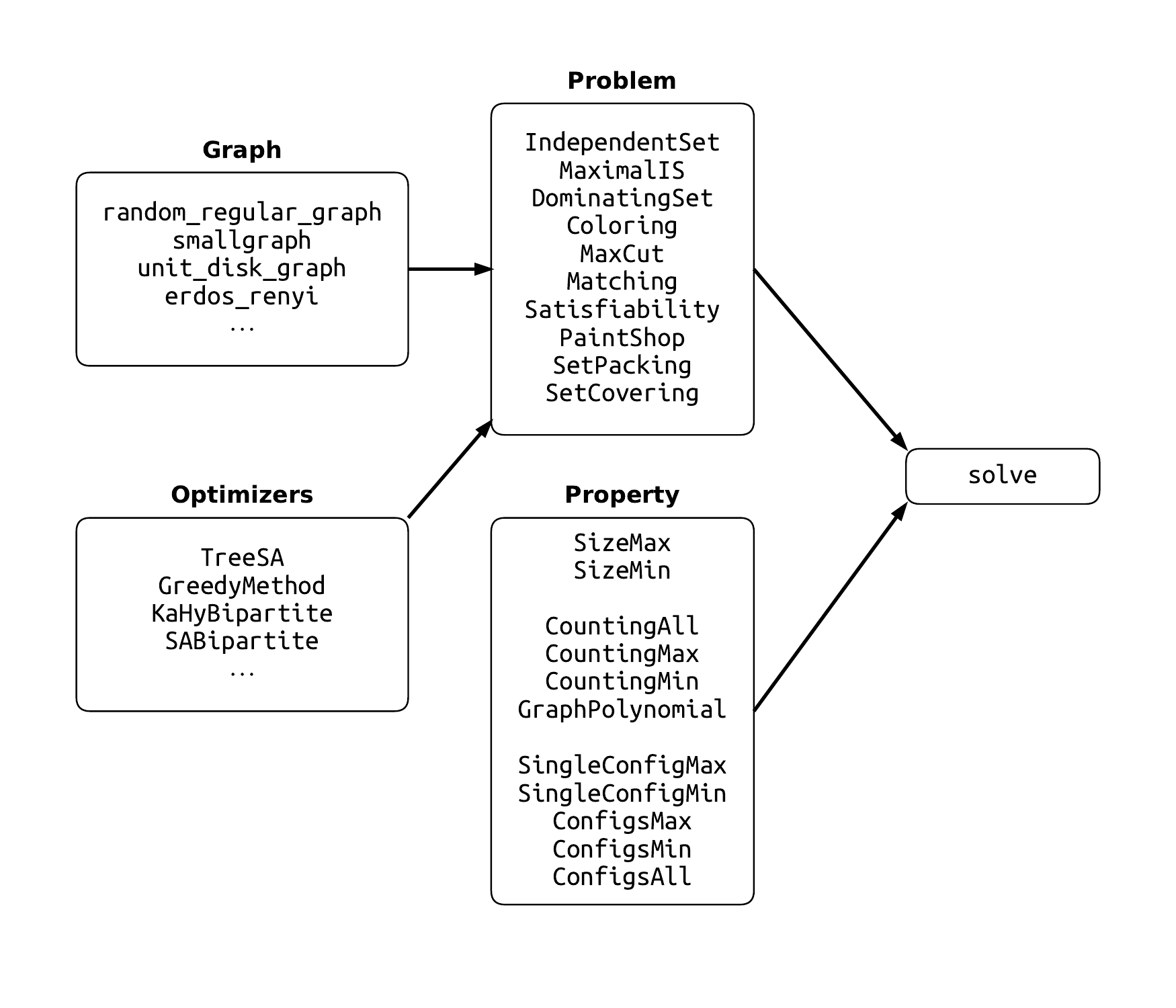}}

The code we will show below is a gist of how the above package was implemented, which is mainly for pedagogical purpose.
It covers most of the topics in the paper without caring much about performance.
It is worth mentioning that this project depends on multiple open source packages in the Julia ecosystem:

\begin{description}
	\item[\href{https://github.com/under-Peter/OMEinsum.jl}{OMEinsum} and \href{https://github.com/TensorBFS/OMEinsumContractionOrders.jl}{OMEinsumContractionOrders}] are packages providing the support for Einstein's (or tensor network) notation and state-of-the-art algorithms for contraction order optimization, which includes the one based on KaHypar+Greedy~\cite{Gray2021, Pan2021} and the one based on local search~\cite{Kalachev2021}.
	\item[\href{https://github.com/TensorBFS/TropicalNumbers.jl}{TropicalNumbers} and \href{https://github.com/TensorBFS/TropicalGEMM.jl}{TropicalGEMM}] are packages providing tropical number and efficient tropical matrix multiplication.
	\item[\href{https://github.com/JuliaGraphs/Graphs.jl}{Graphs}] is a foundational package for graph manipulation in the Julia community.
	\item[\href{https://github.com/JuliaMath/Polynomials.jl}{Polynomials}] is a package providing polynomial algebra and polynomial fitting.
	\item[\href{https://github.com/scheinerman/Mods.jl}{Mods} and the \href{https://github.com/JuliaMath/Primes.jl}{Primes}] package providing finite field algebra and prime number manipulation.
\end{description}

They can be installed in a similar way to \texttt{GenericTensorNetworks}.
After installing the required packages, one can open a Julia REPL, and copy-paste the following code snippet into it.

\begin{lstlisting}[breaklines]
using OMEinsum, OMEinsumContractionOrders
using Graphs
using Random

# generate a random regular graph of size 100, degree 3
graph = (Random.seed!(2); Graphs.random_regular_graph(50, 3))

# generate einsum code, i.e. the labels of tensors
code = EinCode(([minmax(e.src,e.dst) for e in Graphs.edges(graph)]..., # labels for edge tensors
                [(i,) for i in Graphs.vertices(graph)]...), ())        # labels for vertex tensors

# an einsum contraction without a contraction order specified is called `EinCode`,
# an einsum contraction having a contraction order (specified as a tree structure) is called `NestedEinsum`.
# assign each label a dimension-2, it will be used in the contraction order optimization
# `uniquelabels` function extracts the tensor labels into a vector.
size_dict = Dict([s=>2 for s in uniquelabels(code)])
# optimize the contraction order using the `TreeSA` method; the target space complexity is 2^17
optimized_code = optimize_code(code, size_dict, TreeSA())
println("time/space complexity is $(OMEinsum.timespace_complexity(optimized_code, size_dict))")

# a function for computing the independence polynomial
function independence_polynomial(x::T, code) where {T}
	xs = map(getixsv(code)) do ix
        # if the tensor rank is 1, create a vertex tensor.
        # otherwise the tensor rank must be 2, create a bond tensor.
        length(ix)==1 ? [one(T), x] : [one(T) one(T); one(T) zero(T)]
    end
    # both `EinCode` and `NestedEinsum` are callable, inputs are tensors.
	code(xs...)
end

########## COMPUTING THE MAXIMUM INDEPENDENT SET SIZE AND ITS COUNTING/DEGENERACY ###########

# using Tropical numbers to compute the MIS size and the MIS degeneracy.
using TropicalNumbers
mis_size(code) = independence_polynomial(TropicalF64(1.0), code)[]
println("the maximum independent set size is $(mis_size(optimized_code).n)")

# A `CountingTropical` object has two fields, tropical field `n` and counting field `c`.
mis_count(code) = independence_polynomial(CountingTropical{Float64,Float64}(1.0, 1.0), code)[]
println("the degeneracy of maximum independent sets is $(mis_count(optimized_code).c)")

########## COMPUTING THE INDEPENDENCE POLYNOMIAL ###########

# using Polynomial numbers to compute the polynomial directly
using Polynomials
println("the independence polynomial is $(independence_polynomial(Polynomial([0.0, 1.0]), optimized_code)[])")

########## FINDING MIS CONFIGURATIONS ###########

# define the set algebra
struct ConfigEnumerator{N}
    # NOTE: BitVector is dynamic and it can be very slow; check our repo for the static version
    data::Vector{BitVector}
end
function Base.:+(x::ConfigEnumerator{N}, y::ConfigEnumerator{N}) where {N}
    res = ConfigEnumerator{N}(vcat(x.data, y.data))
    return res
end
function Base.:*(x::ConfigEnumerator{L}, y::ConfigEnumerator{L}) where {L}
    M, N = length(x.data), length(y.data)
    z = Vector{BitVector}(undef, M*N)
    for j=1:N, i=1:M
        z[(j-1)*M+i] = x.data[i] .| y.data[j]
    end
    return ConfigEnumerator{L}(z)
end
Base.zero(::Type{ConfigEnumerator{N}}) where {N} = ConfigEnumerator{N}(BitVector[])
Base.one(::Type{ConfigEnumerator{N}}) where {N} = ConfigEnumerator{N}([falses(N)])

# the algebra sampling one of the configurations
struct ConfigSampler{N}
    data::BitVector
end

function Base.:+(x::ConfigSampler{N}, y::ConfigSampler{N}) where {N}  # biased sampling: return `x`
    return x  # randomly pick one
end
function Base.:*(x::ConfigSampler{L}, y::ConfigSampler{L}) where {L}
    ConfigSampler{L}(x.data .| y.data)
end

Base.zero(::Type{ConfigSampler{N}}) where {N} = ConfigSampler{N}(trues(N))
Base.one(::Type{ConfigSampler{N}}) where {N} = ConfigSampler{N}(falses(N))

# enumerate all configurations if `all` is true; compute one otherwise.
# a configuration is stored in the data type of `StaticBitVector`; it uses integers to represent bit strings.
# `ConfigTropical` is defined in `TropicalNumbers`. It has two fields: tropical number `n` and optimal configuration `config`.
# `CountingTropical{T,<:ConfigEnumerator}` stores configurations instead of simple counting.
function mis_config(code; all=false)
    # map a vertex label to an integer
    vertex_index = Dict([s=>i for (i, s) in enumerate(uniquelabels(code))])
    N = length(vertex_index)  # number of vertices
    xs = map(getixsv(code)) do ix
        T = all ? CountingTropical{Float64, ConfigEnumerator{N}} : CountingTropical{Float64, ConfigSampler{N}}
        if length(ix) == 2
            return [one(T) one(T); one(T) zero(T)]
        else
            s = falses(N)
            s[vertex_index[ix[1]]] = true  # one hot vector
            if all
                [one(T), T(1.0, ConfigEnumerator{N}([s]))]
            else
                [one(T), T(1.0, ConfigSampler{N}(s))]
            end
        end
    end
	return code(xs...)
end

println("one of the optimal configurations is $(mis_config(optimized_code; all=false)[].c.data)")

# direct enumeration of configurations can be very slow; please check the bounding version in our Github repo.
println("all optimal configurations are $(mis_config(optimized_code; all=true)[].c)")
\end{lstlisting}

\end{document}

%% file: shared.tex
\pdfoutput=1
\usepackage[english]{babel}
\usepackage[utf8]{inputenc}
\usepackage[T1]{fontenc}
\usepackage{amssymb}
\usepackage{tabularx}
\usepackage{quoting}
\usepackage{upquote}
\usepackage{subcaption}
\usepackage{multicol}
\usepackage{cancel}
\usepackage[framemethod=TikZ]{mdframed}
\usetikzlibrary{shapes}
\usetikzlibrary{snakes}
\usetikzlibrary{shapes.geometric}
\usepackage{wrapfig}
\usepackage{algpseudocode}
\usepackage{rotating}
\usepackage{booktabs}
\usepackage{xcolor}
\makeatletter
\newsavebox{\@brx}
\newcommand{\llangle}[1][]{\savebox{\@brx}{\(\m@th{#1\langle}\)}%
  \mathopen{\copy\@brx\kern-0.5\wd\@brx\usebox{\@brx}}}
\newcommand{\rrangle}[1][]{\savebox{\@brx}{\(\m@th{#1\rangle}\)}%
  \mathclose{\copy\@brx\kern-0.5\wd\@brx\usebox{\@brx}}}
\makeatother
\usepackage{bbm}
\usepackage{jlcode}
\usepackage{graphicx}
\usepackage{amsmath,color}
\usepackage{mathrsfs}
\usepackage{float}
\usepackage[normalem]{ulem}
\usepackage{makecell}
\usepackage{indentfirst}
\usepackage{txfonts}

\makeatletter
\def\parsept#1#2#3{%
    \def\nospace##1{\zap@space##1 \@empty}%
    \def\rawparsept(##1,##2){%
        \edef#1{\nospace{##1}}%
        \edef#2{\nospace{##2}}%
    }%
    \expandafter\rawparsept#3%
}
\makeatother
\DeclareMathAlphabet{\mymathbb}{U}{BOONDOX-ds}{m}{n}
\newcommand{\listingcaption}[1]%
{%
\refstepcounter{lstlisting}\hfill%
Listing \thelstlisting: #1\hfill
}%
\newcolumntype{b}{X}
\newcolumntype{s}{>{\hsize=.7\hsize}X}
\usepackage{listings}
\emergencystretch=\maxdimen
\makeatletter

\def\journal #1, #2, #3, 1#4#5#6{{\sl #1~}{\bf #2}, #3 (1#4#5#6) }

\newcommand{\eqname}[1]{\stepcounter{equation}\tag{\theequation : #1}}

\renewcommand{\H}{{\rm H}}

\renewcommand{\v}[1]{{\bf #1}}

\newcommand{\ra}[1]{\renewcommand{\arraystretch}{#1}}

\newcommand{\Tensors}{{\mathcal{T}}}
\newcommand{\cc}{cc({\mathcal{N}_{\rm IS}(G)})}

\newcommand{\material}[1]{\iffalse[{\bf  \color{cyan}{Material: #1}}]\fi}

\newcounter{example}
\newenvironment{example}[1][]{\refstepcounter{example}\par\medskip
   \noindent \textbf{Example~\theexample. #1} \rmfamily}{\medskip}

\makeatother

\author{
Jin-Guo Liu\thanks{Department of Physics, Harvard University, Cambridge, Massachusetts 02138, USA; QuEra Computing Inc., 1284 Soldiers Field Road, Boston, MA, 02135, USA  
    (\email{jinguoliu@g.harvard.edu}).}
\and {Xun Gao}\thanks{Department of Physics, Harvard University, Cambridge, Massachusetts 02138, USA (\email{xungao@g.harvard.edu}, contributed equally with Jin-Guo Liu to this work.).}
\and Madelyn Cain\footnotemark[2]
\and Mikhail D. Lukin\footnotemark[2]
\and Sheng-Tao Wang\thanks{QuEra Computing Inc., 1284 Soldiers Field Road, Boston, MA, 02135, USA}
}